\documentclass[pdflatex,sn-mathphys-num]{sn-jnl}

\usepackage{graphicx}%
\usepackage{multirow}%
\usepackage{amsmath,amssymb,amsfonts}%
\usepackage{amsthm}%
\usepackage{mathrsfs}%
\usepackage[title]{appendix}%
\usepackage{xcolor}%
\usepackage{textcomp}%
\usepackage{manyfoot}%
\usepackage{booktabs}%
\usepackage{algorithm}%
\usepackage{listings}%
\usepackage{url}
\usepackage[inline]{enumitem}
\usepackage{xcolor}
\usepackage{mathtools}
\usepackage{algorithmic}
\usepackage[disable]{todonotes}
\usepackage{stmaryrd}
\usepackage{varwidth}
\usepackage{multicol}
\usepackage{tikz}
\usepackage{wrapfig}
\usepackage{cleveref}
\usepackage{etoolbox}
\usepackage[T1]{fontenc}

\setenumerate[1]{label={(\arabic*)}}

\usetikzlibrary{arrows,positioning,decorations.markings}
\tikzset{
  varnode/.style={rectangle,outer sep=0mm},
  varnodenoperi/.style={rectangle,outer sep=-1mm},
  ourarrow/.style={>=stealth}, 
  ourarc/.style={>=stealth,thick,arc}}
\usepackage{pgfplots}

\tikzset{
  textnode/.style={},
  varnode2/.style={draw, outer sep=1pt},
  varnode2green/.style={draw,fill=green},
  varnode2red/.style={draw,fill=red},
  varnode2blue/.style={draw,fill=blue},
  varnode/.style={rectangle},
  varnodegreen/.style={rectangle,draw,fill=green},
  varnodered/.style={rectangle,draw,fill=red},
  varnodeblue/.style={rectangle,draw,fill=blue},
  varnodepurple/.style={rectangle,draw,fill=purple},
  varnodeyellow/.style={rectangle,draw,fill=yellow},
  varnodeempty/.style={inner sep=0pt,fill},
  line/.style={=stealth,thick,fill=red},
  matchededge/.style={>=stealth,thick,fill=black},
  unmatchededge/.style={>=stealth,thick,fill=blue},
  arrow/.style={=>stealth,thick}
}

\lstdefinestyle{inline}{%
    basicstyle=\ttfamily%
}

\lstdefinelanguage{isabelle}{ morekeywords={for,context,inductive,begin,locale,fixes,record,type_synonym,definition,fun,function,primrec,where,lemma,theorem,unfolding,by,shows,assumes,and,datatype,using,abbreviation,moreover,have,hence,thus,qed,proof,let,ultimately,show,next,in,if}
    , sensitive=true
    , showstringspaces=false
    , framerule=0pt
    , xleftmargin=2em
    , numbers=left
    , numberstyle=\ttfamily\small
    , firstnumber=1
    , stepnumber=2
    , basicstyle=\ttfamily\small
    , backgroundcolor = \color{white}
    , keywordstyle = {\color{blue}}
    , breaklines=true
    , showspaces=false
    , morecomment=[l]{--}
    , morecomment=[s]{(*}{*)}
    , commentstyle=\color{gray}
    , morestring=[b]"
    , literate={↦}{{$\mapsto$}}{1}
               {∧}{{$\wedge$}}{1}
               {×}{{$\times$}}{1}
               {≡}{{$\equiv$}}{1}
               {∀}{{$\forall$}}{1}
               {∃}{{$\exists$}}{1}
               {\\<and>}{{$\land$}}{1}
        {∈}{{$\in$}}{1} {⇒}{{$\Rightarrow$}}{1} {λ}{{$\lambda$}}{1} {::}{{$::$}}{1}
        {\⊆}{{$\subseteq$}}{1} {\\<subset>}{{$\subset$}}{1} {\\<^sub>m}{{$_m$}}{1} {\\<longleftrightarrow>}{{$\longleftrightarrow$}}{3}
        {\\<pi>}{{$\pi$}}{1} {\\<delta>}{{$\delta$}}{1} {⟦}{{$\llbracket$}}{1} {⟧}{{$\rrbracket$}}{1}
        {⟹}{{$\Longrightarrow$}}{3} {\\<not>}{{$\lnot$}}{1} {\\<le>}{{$\le$}}{1} {\\<rightharpoonup>}{{$\rightharpoonup$}}{2}
        {\\<^sub>\\<V>}{{$_{\mathcal V}$}}{1} {\\<lparr>}{{$\llparenthesis$}}{1} {\\<rparr>}{{$\rrparenthesis$}}{1}
        {\\<leftarrow>}{{$\leftarrow$}}{1} {\\<^sub>\\<O>}{{$_{\mathcal O}$}}{1} {\\<^sub>I}{{$_{\texttt{I}}$}}{1}
        {\\<^sub>G}{{$_{\texttt{G}}$}}{2} {\\<phi>}{{$\varphi$}}{1} {\\<Phi>}{{$\Phi$}}{1} {\\<psi>}{{$\psi$}}{1} {\\<Psi>}{{$\Psi$}}{1}
        {\\<^sub>S}{{$_{\texttt S}$}}{1} {\\<inverse>}{{$^{-1}$}}{1} {\\<^sub>O}{{$_{\texttt O}$}}{1} {\\<^bold>\⋀}{{$\bm\bigwedge$}}{1} {⋀}{{$\bigwedge$}}{1}
        {\\<^bold>\\<or>}{{$\bm\lor$}}{1} {\\<^sub>G}{{$_{\texttt G}$}}{1} {\\<Pi>}{{$\Pi$}}{1} {\\<^sub>I}{{$_{\texttt I}$}}{1} {\≠}{{$\neq$}}{1}
        {\\<bottom>}{{$\bot$}}{1} {\\<^sub>+}{{$_\texttt +$}}{1} {\\<^bold>\\<and>}{{$\bm\land$}}{1} {\\<^bold>\\<not>}{{$\bm\lnot$}}{1}
        {\\<^sub>1}{{$_1$}}{1} {\\<^sub>2}{{$_2$}}{1} {\\<A>}{{$\mathcal A$}}{1} {\\<Turnstile>}{{$\models$}}{2} {\\<^sub>∀}{{$_\forall$}}{1}
        {\\<^sub>0}{{$_0$}}{1} {\\<tau>}{{$\tau$}}{1}  {\\<^sub>\\<Omega>}{{$_\Omega$}}{1} {\\<^sub>V}{{$_V$}}{1} {\\<^bold>\\<Or>}{{$\bm\bigvee$}}{1}
        {\\<^sub>P}{{$_\texttt P$}}{1} {\\<^sub>X}{{$_\texttt X$}}{1} {⟹}{{$\longrightarrow$}}{2} {\\<or>}{{$\lor$}}{1} {\\<^sub>\\<pi>}{{$_\pi$}}{1}
        {\\<^sub>s}{{$_s$}}{1} {\\<^sub>t}{{$_t$}}{1} {\\<^sub>a}{{$_a$}}{1} {\\<^sub>r}{{$_r$}}{1} {\\<^sub>t}{{$_t$}}{1} {\\<^sub>e}{{$_e$}}{1} {\\<^sub>n}{{$_n$}}{1} {\\<^sub>d}{{$_d$}}{1} {\\<^sub>i}{{$_i$}}{1} {\\<^sub>v}{{$_v$}}{1} {\\<^sub>j}{{$_j$}}{1} {\\<^sub>b}{{$_b$}}{1} {\∩}{{$\cap$}}{1} {\\<union>}{{$\cup$}}{1} {\\<Union>}{{$\bigcup$}}{1} {\\<^sup>c\\<TTurnstile>\\<^sub>=}{{${}^c\models_=$}}{1}
        {\\<open>}{{<}}{1} {\\<close>}{{>}}{1} {\\<langle>}{{$\langle$}}{1} {\\<rangle>}{{$\rangle$}}{1} {\\<ge>}{{$\ge$}}{1} {\\<^sup>-\\<^sup>1\\<^sub>C}{{$^{\texttt{-1}}_\texttt{C}$}}{2} {\\<^sup>+\\<^sub>C}{{$^\texttt{+}_\texttt{C}$}}{1} {\\<circ>\\<^sub>C}{{$\circ^\texttt C$}}{1} {\\<top>\\<^sub>C}{{$\top_\texttt{C}$}}{2} {\\<bottom>\\<^sub>C}{{$\bot_\texttt{C}$}}{2} {\\<not>\\<^sub>C}{{$\neg_\texttt{C}$}}{2}
        {\\<squnion>\\<^sub>C}{{$\sqcup_\texttt{C}$}}{2} {\\<sqinter>\\<^sub>C}{{$\sqcap_\texttt{C}$}}{2} {\∃\\<^sub>C}{{$\exists_\texttt{C}$}}{2}  {∀\\<^sub>C}{{$\forall_\texttt{C}$}}{2} {=\\<^sub>C}{{$=_\texttt{C}$}}{2} {\\<sigma>}{{$\sigma$}}{2} {\∉}{{$\notin$}}{1} {⊕}{{$\oplus$}}{1} {\\<nexists>}{{$\nexists$}}{1} {\\<setminus>}{{$\setminus$}}{1} {`}{{`}}{1}
}

\begin{document}

\title[A Formal Correctness Proof of Edmonds' Blossom Shrinking Algorithm]{A Formal Correctness Proof of Edmonds' Blossom Shrinking Algorithm}


\author{Mohammad Abdulaziz\textsuperscript{0000-0002-8244-518X}}\email{mohammad.abdulaziz@kcl.ac.uk}
\affil{King's College London, \city{London}, \country{United Kingdom}}
\author{Kurt Mehlhorn\textsuperscript{0000-0003-4020-4334}}\email{mehlhorn@mpi-inf.mpg.de}
\affil{Max Planck Institute for Informatics, \city{Saarbrücken}, \country{Germany}}
\abstract{We present the first formal correctness proof of Edmonds' blossom shrinking algorithm for maximum cardinality matching in general graphs. We focus on formalising the mathematical structures and properties that allow the algorithm to run in worst-case polynomial running time. We formalise Berge's lemma, blossoms and their properties, and a mathematical model of the algorithm, showing that it is totally correct. We provide the first detailed proofs of many of the facts underlying the algorithm's correctness.}

\maketitle

\providecommand{\insts}{}
\renewcommand{\insts}{\ensuremath{\Delta}}
\providecommand{\inst}{\ensuremath{\tvsal}}
\newcommand{\act}{\ensuremath{\pi}}
\newcommand{\asarrow}[1]{\vec{#1}}
\renewcommand{\vec}[1]{\overset{\rightarrow}{#1}}
\newcommand{\as}{\ensuremath{\vec{{\act}}}}

\newcommand{\etc}{\textit{etc.}}
\newcommand{\versus}{\textit{vs.}}

\newcommand{\Ie}{I.e.}
\newcommand{\eg}{e.g.}
\newcommand*{\ie}{i.e.\ }
\newcommand{\abziz}[1]{\textcolor{brown}{#1}}
\newcommand{\sublist}[2]{ \ensuremath{#1} \preceq\!\!\!\raisebox{.4mm}{\ensuremath{\cdot}}\; \ensuremath{#2}}
\newcommand{\subscriptsublist}[2]{\ensuremath{#1}\preceq\!\raisebox{.05mm}{\ensuremath{\cdot}}\ensuremath{#2}}
\newcommand{\PLS}{\Pi^\preceq\!\raisebox{1mm}{\ensuremath{\cdot}}}
\newcommand{\PLScharles}{\Pi^d}
\newcommand{\execname}{\mathsf{ex}}
\newcommand{\IndHyp}{\mathsf{IH}}
\newcommand{\exec}[2]{#2(#1)}

\newcommand{\ancestorssymbol}{\textsf{\upshape ancestors}}
\newcommand{\ancestors}{\ancestorssymbol}
\newcommand{\satpreas}[2]{\ensuremath{sat_precond_as(s, \as)}}
\newcommand{\proj}[2]{\ensuremath{#1{\downharpoonright}_{#2}}}
\newcommand{\dep}[3]{\ensuremath{#2 {\rightarrow} #3}}
\newcommand{\deptc}[3]{\ensuremath{#2 {\rightarrow^+} #3}}
\newcommand{\negdep}[3]{\ensuremath{#2 \not\rightarrow #3}}
\newcommand{\leavessymbol}{\textsf{\upshape leaves}}
\newcommand{\leaves}{\leavessymbol}

\newcommand{\childrensymbol}{\textsf{\upshape children}}
\newcommand{\children}[2]{\mathcal{\childrensymbol}_{#2}(#1)}
\newcommand{\succsymbol}{\textsf{\upshape succ}}
\newcommand{\succstates}[2]{\succsymbol(#1, #2)}
\newcommand{\concat}{\#}
\newcommand{\RG}{\cite{Rintanen:Gretton:2013}\ }
\newcommand{\cupdot}{\charfusion[\mathbin]{\cup}{\cdot}}
\newcommand{\bigcupdot}{\charfusion[\mathop]{\bigcup}{\cdot}}
\newcommand{\cuparrow}{\charfusion[\mathbin]{\cup}{{\raisebox{.5ex} {\smathcalebox{.4}{\ensuremath{\leftarrow}}}}}}
\newcommand{\bigcuparrow}{\charfusion[\mathop]{\bigcup}{\leftarrow}}
\newcommand{\finiteunion}{\cuparrow}
\newcommand{\finitemap}{\ensuremath{\sqsubseteq}}
\newcommand{\dgraph}{dependency graph}
\newcommand{\domain}[1]{{\sc #1}}
\newcommand{\solver}[1]{{\sc #1}}
\providecommand{\problem}[1]{\domain{#1}}
\renewcommand{\v}{\ensuremath{\mathit{v}}}
\providecommand{\vs}[1]{\domain{#1}}
\renewcommand{\vs}{\ensuremath{\mathit{vs}}}
\newcommand{\VS}{\ensuremath{\mathit{VS}}}
\newcommand{\Aut}{\ensuremath{\mathit{Aut}}}
\newcommand{\Inst}[2]{\ensuremath{\mathit{#2 \rightarrow_{#1} #1}}}
\newcommand{\Image}{\ensuremath{\mathit{Im}}}
\newcommand{\Img}[2]{\protect{#1 \llparenthesis #2 \rrparenthesis}}
\newcommand{\SND}{\ensuremath{\mathit{\pi_2}}}
\newcommand{\FST}{\ensuremath{\mathit{\pi_1}}}
\newcommand{\tvsal}{{\pitchfork}}
\newcommand{\nauty}{CGIP}

\newcommand{\pwinter}{\ensuremath{\mathit{\bigcap_{pw}}}}

\newcommand{\dom}{\ensuremath{\mathit{\mathcal{D}}}}
\newcommand{\codom}{\ensuremath{\mathcal{R}}}

\newcommand{\map}{\ensuremath{\mathit{map}}}
\newcommand{\BIJEC}{\ensuremath{\mathit{bij}}}
\newcommand{\INJ}{\ensuremath{\mathit{inj}}}
\newcommand{\funion}{\ensuremath{\overset{\leftarrow}{\cup}}}

\newcommand{\ifnew}{\mbox{\upshape \textsf{if}}}
\newcommand{\thennew}{\mbox{\upshape \textsf{then}}}
\newcommand{\elsenew}{\mbox{\upshape \textsf{else}}}
\newcommand{\choice}{\mbox{\upshape \textsf{ch}}}
\newcommand{\arbchoice}{\mbox{\upshape \textsf{arb}}}
\newcommand{\acycchoice}{\mbox{\upshape \textsf{ac}}}
\newcommand{\cycchoice}{\mbox{\upshape \textsf{cyc}}}
\newcommand{\filter}{\ensuremath{\mathit{FIL}}}
\newcommand{\probset}{\ensuremath{\boldsymbol \Pi}}
\newcommand{\probleq}{\ensuremath{\leq_\Pi}}
\newcommand{\CommVar}{\ensuremath{\bigcap_\v} }
\newcommand{\quotfun}{\ensuremath{ \mathcal{Q}}}

\newcommand{\apre}{\mbox{\upshape \textsf{pre}}}
\newcommand{\aeff}{\mbox{\upshape \textsf{eff}}}
\newcommand{\problist}{\ensuremath \probset}
\newcommand{\cat}{{\frown}}
\newcommand{\probproj}[2]{{#1}{\downharpoonright}^{#2}}
\newcommand{\preced}{\mathbin{\rotatebox[origin=c]{180}{\ensuremath{\rhd}}}}
\newcommand{\perm}{\ensuremath{\sigma}}
\newcommand{\invariant}[2]{\ensuremath{\mathit{inv({#1},{#2})}}}
\newcommand{\invstates}[1]{\ensuremath{\mathit{inv({#1})}}}
\newcommand{\probss}[1]{{\mathcal S}(#1)}
\newcommand{\parChildRel}[3]{\ensuremath{\negdep{#1}{#2}{#3}}}
\newcommand{\asessymbol}{\ensuremath{\mathbb{A}}}
\newcommand{\ases}[1]{{#1}^*}
\newcommand{\uniStates}{\ensuremath{\mathbb{U}}}
\newcommand{\recurrenceDiam}{\ensuremath{\mathit{rd}}}
\newcommand{\recurrenceAcycDiamfun}{\ensuremath{\mathit{{\mathfrak A}}}}
\newcommand{\recurrenceDiamfun}{\ensuremath{\mathit{\mathfrak R}}}
\newcommand{\traversalDiam}{\ensuremath{\mathit{td}}}
\newcommand{\traversalDiamfun}{\ensuremath{\mathit{\mathfrak T}}}
\newcommand{\isPrefix}[2]{\ensuremath{#1 \preceq #2}}
\providecommand{\path}{\ensuremath{\gamma}}
\newcommand{\aspath}{\ensuremath{\vec{\path}}}
\newcommand{\n}{\textsf{\upshape n}}
\providecommand{\graph}{}
\renewcommand{\graph}{{\cal G}}
\newcommand{\undirgraph}{{\cal G}}

\renewcommand{\ss}{\ensuremath{\state s}}
\newcommand{\slist}{\ensuremath{\vec{\mbox{\upshape \textsf{ss}}}}}
\newcommand{\sll}{\ensuremath{\vec{\state}}}
\newcommand{\listset}{\mbox{\upshape \textsf{set}}}
\newcommand{\asset}{\ensuremath{\mathit{K}}}
\newcommand{\aslist}{\ensuremath{\mathit{\overset{\rightarrow}{\gamma}}}}
\newcommand{\head}{\mbox{\upshape \textsf{hd}}}
\renewcommand{\max}{\textsf{\upshape max}}
\newcommand{\argmax}{\textsf{\upshape argmax}}
\newcommand{\argmin}{\textsf{\upshape argmin}}
\renewcommand{\min}{\textsf{\upshape min}}
\newcommand{\bool}{\mbox{\upshape \textsf{bool}}}
\newcommand{\last}{\mbox{\upshape \textsf{last}}}
\newcommand{\front}{\mbox{\upshape \textsf{front}}}
\newcommand{\rot}{\mbox{\upshape \textsf{rot}}}
\newcommand{\stuff}{\mbox{\upshape \textsf{intlv}}}
\newcommand{\tail}{\mbox{\upshape \textsf{tail}}}
\newcommand{\ngrtoas}{\ensuremath{\mathit{\as_{\graph_\mathbb{N}}}}}
\newcommand{\vsfun}{\mbox{\upshape \textsf{vs}}}
\newcommand{\inits}{\mbox{\upshape \textsf{init}}}
\newcommand{\satprecondas}{\mbox{\upshape \textsf{sat-pre}}}
\newcommand{\remcondlessact}{\mbox{\upshape \textsf{rem-condless}}}
\providecommand{\state}{}
\renewcommand{\state}{x}
\newcommand{\statea}{x}
\newcommand{\stateb}{y}
\newcommand{\statec}{z}
\newcommand{\fals}{\mbox{\upshape \textsf{F}}}
\newcommand{\indices}{\ensuremath{V}}
\newcommand{\edges}{\ensuremath{E}}
\newcommand{\vertices}{\ensuremath{V}}
\newcommand{\listtype}{\mbox{\upshape \textsf{list}}}
\newcommand{\settype}{\mbox{\upshape \textsf{set}}}
\newcommand{\acttype}{\mbox{\upshape \textsf{action}}}
\newcommand{\graphtype}{\mbox{\upshape \textsf{graph}}}
\newcommand{\projfun}[2]{\ensuremath{\Delta_{#1}^{#2}}}
\newcommand{\snapfun}[2]{\ensuremath{\Sigma_{#1}^{#2}}}
\newcommand{\RDfun}[1]{\ensuremath{{\mathcal R}_{#1}}}
\newcommand{\elldbound}[1]{\ensuremath{{\mathcal LS}_{#1}}}
\newcommand{\distinct}{\textsf{\upshape distinct}}
\newcommand{\ddistinct}{\mbox{\upshape \textsf{ddistinct}}}
\newcommand{\simple}{\mbox{\upshape \textsf{simple}}}

\newcommand{\reachable}[3]{\ensuremath{{#1}\rightsquigarrow{#3}}}

\newcommand{\Omit}[1]{}

\newcommand{\charles}[1]{\textcolor{red}{#1}}

\newcommand{\negreachable}[3]{\ensuremath{{#2}\not\rightsquigarrow{#3}}}
\newcommand{\wdiam}[2]{{#1}^{#2}}
\newcommand{\dsnapshot}[2]{\Delta_{#1}}
\newcommand{\ellsnapshot}[2]{{\mathcal L}_{#1}}

\newcommand{\snapshotsymbol}{|\kern-.7ex\raise.08ex\hbox{\scalebox{0.7}{$\bullet$}}}
\newcommand{\snapshot}[2]{\ensuremath{\mathrel{#1\snapshotsymbol_{#2}}}}
\newcommand{\vstype}{\texttt{\upshape VS}}
\newcommand{\vtype}{{\scriptsize \ensuremath{\dom(\delta)}}}
\newcommand{\Balgo}{{\mbox{\textsc{Hyb}}}}
\newcommand{\ssgraph}[1]{\graph_\ss}
\newcommand{\agree}{\textsf{\upshape agree}}
\newcommand{\ck}{\ensuremath{\texttt{ck}}}
\newcommand{\lk}{\ensuremath{\texttt{lk}}}
\newcommand{\gr}{\ensuremath{\texttt{gr}}}
\newcommand{\gk}{\ensuremath{\texttt{gk}}}
\newcommand{\CK}{\ensuremath{\texttt{CK}}}
\newcommand{\LK}{\ensuremath{\texttt{LK}}}
\newcommand{\GR}{\ensuremath{\texttt{GR}}}
\newcommand{\GK}{\ensuremath{\texttt{GK}}}
\newcommand{\safe}{\ensuremath{\texttt{s}}}

\newcommand{\derivname}{\ensuremath{\partial}}
\newcommand{\deriv}[3]{\ensuremath{\derivname(#1,#2,#3)}}
\newcommand{\derivabbrev}[3]{\ensuremath{{\partial(#1,#2)}}}
\newcommand{\subsetoracle}{\ensuremath{ \Omega}}
\newcommand{\Aalgo}{{\mbox{\textsc{Pur}}}}
\newcommand{\Sname}{\textsf{\upshape S}}
\newcommand{\Sbrace}[1]{\Sname\langle#1\rangle}
\newcommand{\SalgoName}{\Sname_{\textsf{\upshape max}}}
\newcommand{\Salgo}[1]{\SalgoName\langle#1\rangle}

\newcommand{\WLPname}{{\mbox{\textsc{wlp}}}}
\newcommand{\WLPbrace}[1]{\WLPname\langle#1\rangle}
\newcommand{\WLPalgoName}{\WLPname_{\textsf{\upshape max}}}
\newcommand{\WLP}[1]{\WLPalgoName\langle#1\rangle}

\newcommand{\Nname}{\ensuremath{\textsf{\upshape N}}}
\newcommand{\Nbrace}[1]{\Nname\langle#1\rangle}
\newcommand{\NalgoName}{\Nname{_{\textsf{\upshape sum}}}}
\newcommand{\Nalgobrace}[1]{\NalgoName\langle#1\rangle}

\newcommand{\acycNname}{\widehat{\textsf{\upshape N}}}
\newcommand{\acycNbrace}[1]{\acycNname\langle#1\rangle}
\newcommand{\acycNalgoName}{\acycNname{_{\textsf{\upshape sum}}}}
\newcommand{\acycNalgobrace}[1]{\acycNalgoName\langle#1\rangle}

\newcommand{\Mname}{\ensuremath{\textsf{\upshape M}}}
\newcommand{\Mbrace}[1]{\Mname\langle#1\rangle}
\newcommand{\MalgoName}{\Mname{_{\textsf{\upshape sum}}}}
\newcommand{\Malgobrace}[1]{\MalgoName\langle#1\rangle}
\newcommand{\cardinality}[1]{{\ensuremath{|#1|}}}
\newcommand{\length}[1]{\cardinality{#1}}
\newcommand{\basecasefun}{\ensuremath{b}}
\newcommand{\Basecasefun}{\ensuremath{\mathcal B}}

\newcommand{\edgegen}{\ensuremath{e}}
\newcommand{\vertexgen}{\ensuremath{u}}
\newcommand{\vertexa}{{\ensuremath{\vertexgen_1}}}
\newcommand{\vertexb}{{\ensuremath{\vertexgen_2}}}
\newcommand{\vertexc}{{\ensuremath{\vertexgen_3}}}
\newcommand{\vertexd}{{\ensuremath{\vertexgen_4}}}
\newcommand{\vertexe}{{\ensuremath{\vertexgen_5}}}
\newcommand{\vertexf}{{\ensuremath{\vertexgen_6}}}
\newcommand{\vertexg}{{\ensuremath{\vertexgen_7}}}
\newcommand{\vertexh}{{\ensuremath{\vertexgen_8}}}
\newcommand{\vertexi}{{\ensuremath{\vertexgen_9}}}
\newcommand{\vertexj}{\ensuremath{\vertexgen_{11}}}
\newcommand{\vertexk}{\ensuremath{\vertexgen_{12}}}
\newcommand{\vertexl}{\ensuremath{\vertexgen_{12}}}
\newcommand{\vertexm}{\ensuremath{\vertexgen_{13}}}
\newcommand{\vertexn}{\ensuremath{\vertexgen_{14}}}
\newcommand{\vertexo}{\ensuremath{\vertexgen_{15}}}
\newcommand{\vertexsetgen}{\ensuremath{\mathit{us}}}
\newcommand{\vertexseta}{\vertexsetgen_1}
\newcommand{\vertexsetb}{\vertexsetgen_2}
\newcommand{\labelsymbol}{\ensuremath{l}}
\newcommand{\labelfun}{\ensuremath{\mathcal{L}}}
\newcommand{\DAG}{\ensuremath{A}}
\newcommand{\NalgoNameN}{{\ensuremath{\NalgoName_{\mathbb{N}}}}}
\newcommand{\NnameN}{\ensuremath{\Nname_\mathbb{N}}}
\newcommand{\replaceprojsinglename}{\raisebox{-0.3mm} {\scalebox{0.7}{\textpmhg{H}}}}
\newcommand{\replaceprojsingle}[3] {{ #2} \underset {#1} {\raisebox{-0.3mm} {\scalebox{0.7}{\textpmhg{H}}}}  #3}
\newcommand{\HOLreplaceprojsingle}[1]{\underset {#1} {\raisebox{-0.3mm} {\scalebox{0.7}{\textpmhg{H}}}}}

\newcommand{\lotus}{{\scalebox{0.6}{\includegraphics{lotus.pdf}}}}
\newcommand{\invlotus}{\mathbin{\rotatebox[origin=c]{180}{$\lotus$}}}
\newcommand{\clique}{\ensuremath{K}}
\newcommand{\partition}{\ensuremath{\vs_{1..n}}}
\newcommand{\partitiontype}{\ensuremath{\vstype_{1..n}}}
\newcommand{\vtxpartition}{\ensuremath{P}}

\newcommand{\traversalDiamAlgo}{{\mbox{\textsc{TravDiam}}}}
\newcommand{\prefix}{\textsf{\upshape pfx}}
\newcommand{\powerset}{\mathbb{P}}
\newcommand{\postfix}{\textsf{\upshape sfx}}
\newcommand{\dfunproj}{\ensuremath{{\mathfrak D}}}
\newcommand{\dfunsnap}{\ensuremath{{\textgoth D}}}
\newcommand{\ellfunproj}{\ensuremath{\mathfrak L}}
\newcommand{\ellfunsnap}{\ensuremath{\textgoth L}}
\newcommand{\cycle}{\ensuremath{C}}
\newcommand{\petal}{\ensuremath{\eta}}
\renewcommand{\prod}{\ensuremath{{{{{\mathlarger{\mathlarger {{\mathlarger {\Pi}}}}}}}}}}
\newcommand{\sccset}{{\ensuremath{SCC}}}
\newcommand{\scc}{{\ensuremath{scc}}}
\newcommand{\negate}[1]{\overline{#1}}
\newcommand{\setofsets}{\ensuremath{S}}
\newcommand{\group}{\ensuremath{\cal \Gamma}}
\newcommand{\neededvars}{{\cal N}}
\newcommand{\sspace}{\mbox{\upshape \textsf{sspc}}}
\newcommand{\tip}{\ensuremath{t}}
\newcommand{\vara}{\ensuremath{\v_1}}
\newcommand{\varb}{\ensuremath{\v_2}}
\newcommand{\varc}{\ensuremath{\v_3}}
\newcommand{\vard}{\ensuremath{\v_4}}
\newcommand{\vare}{\ensuremath{\v_5}}
\newcommand{\varf}{\ensuremath{\v_6}}
\newcommand{\varg}{\ensuremath{\v_7}}
\newcommand{\varh}{\ensuremath{\v_8}}
\newcommand{\vari}{\ensuremath{\v_9}}
\newcommand{\acta}{\ensuremath{\act_1}}
\newcommand{\actb}{\ensuremath{\act_2}}
\newcommand{\actc}{\ensuremath{\act_3}}
\newcommand{\actd}{\ensuremath{\act_4}}
\newcommand{\acte}{\ensuremath{\act_5}}
\newcommand{\actf}{\ensuremath{\act_6}}
\newcommand{\actg}{\ensuremath{\act_7}}
\newcommand{\acth}{\ensuremath{\act_8}}
\newcommand{\acti}{\ensuremath{\act_9}}

\tikzset{dots/.style args={#1per #2}{line cap=round,dash pattern=on 0 off #2/#1}}
\providecommand{\moham}[1]{\fbox{{\bf \@Mohammad: }#1}}
\newcommand{\TDbound}{{\mbox{\textsc{Arb}}}}
\newcommand{\expbound}{{\mbox{\textsc{Exp}}}}
\newcommand{\sasdom}{\expbound}
\newcommand{\cardfun}{\ensuremath{\mathbb{C}}}
\newcommand{\AGNa}{AGN1}
\newcommand{\AGNb}{AGN2}
\newcommand{\reset}{{\ensuremath{reset}}}

\newcommand{\matching}{{\cal M}}
\newcommand{\BlossomAlg}{{\mbox{\textsc{Find\_Max\_Matching}}}}
\newcommand{\AugPathAlg}{{\mbox{\textsc{Aug\_Path\_Search}}}}
\newcommand{\BlossomOrAugPath}{{\mbox{\textsc{Blossom\_Search}}}}

\renewcommand{\vertices}{\ensuremath{\mathcal{V}}}
\newcommand{\lparty}{\ensuremath{V}}
\newcommand{\rparty}{\ensuremath{U}}
\newcommand{\lperm}{\ensuremath{\sigma}}
\newcommand{\rperm}{\ensuremath{\pi}}
\renewcommand{\vertexgen}{\ensuremath{v}}
\newcommand{\lvertexgen}{\ensuremath{v}}
\newcommand{\lvertexa}{\lvertexgen_1}
\newcommand{\lvertexb}{\lvertexgen_2}
\newcommand{\lvertexc}{\lvertexgen_3}
\newcommand{\lvertexd}{\lvertexgen_4}
\newcommand{\lvertexe}{\lvertexgen_5}
\newcommand{\lvertexf}{\lvertexgen_6}
\newcommand{\rvertexgen}{\ensuremath{u}}
\newcommand{\rvertexa}{\rvertexgen_1}
\newcommand{\rvertexb}{\rvertexgen_2}
\newcommand{\rvertexc}{\rvertexgen_3}
\newcommand{\rvertexd}{\rvertexgen_4}
\newcommand{\rvertexe}{\rvertexgen_5}
\newcommand{\rvertexf}{\rvertexgen_6}
\newcommand{\lorder}{\ensuremath{\pi}}
\newcommand{\rorder}{\ensuremath{\sigma}}
\newcommand{\rank}{\textit{online-match}}
\newcommand{\neighb}[2]{\ensuremath{{N_{#1} ({#2})}}}
\newcommand{\shiftsto}{\textit{shifts-to}}
\newcommand{\zig}{\textit{zig}}
\newcommand{\zag}{\textit{zag}}
\newcommand{\lpartyperm}{\ensuremath{\lparty'}}

\newcommand{\isaname}[1]{\emph{#1}}

\newcommand{\nth}[2]{#1[#2]}
\newcommand*{\uniform}{\mathcal{U}}
\newcommand*{\rankingprob}{\textit{RANKING}}
\newcommand*{\perms}{\mathcal{S}}
\newcommand*{\bernoulli}{\mathbb{I}}
\DeclarePairedDelimiter{\card}{|}{|}
\providecommand{\hide}[1]{}
\providecommand{\cproblem}{\Pi}
\providecommand{\creduction}{f}
\newcommand{\workpackagealgos}{WP 1}
\newcommand{\workpackagereductions}{WP 2}
\newcommand{\workpackageMV}{WP 3}
\newcommand{\Bone}[1]{}
\newcommand{\Btwo}[1]{}
\newcommand{\objMV}{1}
\newcommand{\objAlgos}{2}
\newcommand{\isa}[1]{\texttt{#1}}
\newcommand{\isadigit}[1]{\texttt{#1}}
\newcommand{\edge}{\ensuremath{e}}
\newcommand{\ihvar}[1]{\overline{#1}}
\newcommand{\nextIterVar}[1]{{#1}'}
\newcommand{\varAtLine}[2]{{#1}^{#2}}
\newcommand{\emptylist}{\ensuremath{\emptyset}}
\newcommand{\stem}{\ensuremath{\text{stem}}}
\newcommand{\reverse}{\ensuremath{\text{rev}}}

\makeatletter
\newtheoremstyle{implication}
  {4pt}
  {4pt}
  {\addtolength{\@totalleftmargin}{3.5em}
   \addtolength{\linewidth}{0em}
   \parshape 1 0em \linewidth}
  {}
  {}
  {}
  {.5em}
  {}
\makeatother

\makeatletter
\newtheoremstyle{indented}
  {4pt}
  {4pt}
  {\addtolength{\@totalleftmargin}{3.5em}
   \addtolength{\linewidth}{0em}
   \parshape 1 0em \linewidth}
  {}
  {\it}
  {.}
  {.5em}
  {}
\makeatother

\makeatletter
\newtheoremstyle{indentedtwice}
  {4pt}
  {4pt}
  {\addtolength{\@totalleftmargin}{3.5em}
   \addtolength{\linewidth}{0em}
   \parshape 1 0em \linewidth}
  {2ex}
  {\it}
  {.}
  {.5em}
  {}
\makeatother

\makeatletter
\newtheoremstyle{indentedthrice}
  {4pt}
  {4pt}
  {\addtolength{\@totalleftmargin}{3.5em}
   \addtolength{\linewidth}{0em}
   \parshape 1 0em \linewidth}
  {4ex}
  {\it}
  {.}
  {.5em}
  {}
\makeatother

\makeatletter
\newtheoremstyle{indentedfour}
  {4pt}
  {4pt}
  {\addtolength{\@totalleftmargin}{3.5em}
   \addtolength{\linewidth}{0em}
   \parshape 1 0em \linewidth}
  {6ex}
  {\it}
  {.}
  {.5em}
  {}
\makeatother

\newtheorem{mynote}{Note}
\newtheorem{mythm}{Theorem}
\newtheorem{mylem}{Lemma}
\newtheorem{mycor}{Corollary}
\newtheorem{myprop}{Proposition}
\newtheorem{myspec}{Specification}
\newtheorem{myinvar}{Invariant}
\theoremstyle{definition}
\newtheorem{mydef}{Definition}
\theoremstyle{remark}
\newtheorem{myremark}{Remark}
\theoremstyle{implication}
\newtheorem{myleft}{$(\Leftarrow)$}
\renewcommand{\themyleft}{}
\newtheorem{myright}{$(\Rightarrow)$}
\renewcommand{\themyright}{}
\theoremstyle{indented}
\newtheorem{mycase}{Case}
\renewcommand{\themycase}{\arabic{mycase}}
\theoremstyle{indentedtwice}
\newtheorem{mysubcase}{Case}
\renewcommand{\themysubcase}{\arabic{mycase}.\roman{mysubcase}}
\theoremstyle{indentedthrice}
\newtheorem{mysubsubcase}{Case}
\renewcommand{\themysubsubcase}{\arabic{mycase}.\roman{mysubcase}.\alph{mysubsubcase}}
\theoremstyle{indentedfour}
\newtheorem{mysubsubsubcase}{Case}
\renewcommand{\themysubsubsubcase}{\arabic{mycase}.\roman{mysubcase}.\alph{mysubsubcase}.\Roman{mysubsubsubcase}}

\let\oldcomplement\complement
\renewcommand{\complement}[1]{\ensuremath{{#1}^\oldcomplement}}

\makeatletter
\@addtoreset{mycase}{mythm}
\@addtoreset{mycase}{mylem}
\@addtoreset{mycase}{myprop}
\@addtoreset{mycase}{mycor}
\@addtoreset{mycase}{proof}
\@addtoreset{mycase}{myleft}
\@addtoreset{mycase}{myright}

\@addtoreset{mysubcase}{mycase}
\@addtoreset{mysubsubcase}{mysubcase}
\@addtoreset{mysubsubsubcase}{mysubsubcase}

\makeatother

\tableofcontents

\section{Introduction}
\label{sec:intro}

Maximum cardinality matching is a basic problem in computer science, operational research, graph theory, and combinatorial optimisation.
In this problem, given an undirected graph, one is to find the largest, in terms of cardinality, subset of vertex disjoint edges of that graph.

We describe the first formal functional correctness proof of Edmonds' blossom shrinking algorithm~\cite{edmond1965blossom}, which is an algorithm to solve the maximum cardinality matching problem in general graphs.
We do the proof in the theorem prover Isabelle/HOL~\cite{IsabelleHOLRef}.
Developing a formal correctness proof for this algorithm presents substantial challenges.
First, the correctness argument depends on substantial graph theory, including results like Berge's lemma~\cite{BergesLemma}.
Second, it includes reasoning about graph contractions, with complex case analyses associated with this type of reasoning.
Third, the algorithm's core procedure is an involved iterative search procedure that builds a complex forest data structure.
Proving its total correctness depends on a large number of complex loop invariants and on the construction of a complex certificate.

Our contributions here include:
\begin{enumerate}
\item developing substantial formal libraries for undirected graphs, including reasoning principles for performing mathematical induction and treating connected components; alternating paths, a necessary concept for reasoning about matchings and matching algorithms; and matching theory, including Berge's lemma,
\item methodology-wise, we use Isabelle/HOL's function package to model all iterative computation; Isabelle/HOL's locales to structure our proofs; and Isabelle/HOL's classical reasoning to automate much of our proofs, showing that standard tools of Isabelle/HOL already suffice to perform reasoning about some of the most complex algorithms in a relatively elegant fashion, and 
\item mathematically, despite the existence of many established expositions~\cite{edmond1965blossom,LEDAbook,KorteVygenOptimisation,schrijverBook},  we \begin{enumerate*} \item provide the first complete case analyses of two central results: the decades old Berge's lemma and the fact that \emph{blossom shrinking} preserves \emph{augmenting path} existence; and \item provide the first complete list of invariants and the first detailed correctness proof of the core search procedure.
\end{enumerate*}
\end{enumerate}
We note that parts of this work were presented in a preliminary form in an earlier invited conference paper~\cite[Section 5]{BlossomIsabelle}, which had only a proof of partial correctness.
We also note that our formalisation is not based on any specific source.
However, we point out similarities to existing sources when they arise.

The structure of the paper will be as follows: we first discuss the necessary background notions.
Then we have three technical sections, each dedicated to a major algorithmic part and to the mathematics behind its correctness.
In each of those sections, in addition to presenting the formalisation, we present our own informal proofs, which we believe provide necessary mathematical insights for those who are interested in the formalisation, as well as those who are interested in the algorithm's correctness generally.
Then we finish with a discussion section.



\renewcommand{\path}{\ensuremath{\gamma}}

\begin{algorithm}[H]
    \caption{$\BlossomAlg(\graph)$}\label{alg:Blossom}
    \label{alg:algorithm-label}
    \begin{algorithmic}[1]
      \STATE $\matching := \emptyset$
      \STATE $\path := \AugPathAlg(\graph,\matching)$
      \WHILE{$\gamma$ is an augmenting path}
        \STATE {$\matching := \matching \oplus \path$}
        \STATE {$\path:= \AugPathAlg(\graph,\matching)$}
      \ENDWHILE
      \RETURN $\matching$
    \end{algorithmic}
\end{algorithm}

\begin{figure}[H]
  \centering
  \begin{tikzpicture}[rotate=0,scale=0.8]
    \node (v1) at (0,1) [varnode] {\scriptsize $\vertexa$ } ;
    \node (v2) at (0,-1) [varnode] {\scriptsize $\vertexb$ } ;
    \node (v3) at (1,0) [varnode] {\scriptsize $\vertexc$ } ;
    \node (v4) at (2,0) [varnode] {\scriptsize $\vertexd$ } ;
    \node (v5) at (3,0) [varnode] {\scriptsize $\vertexe$ } ;
    \node (v6) at (4,1) [varnode] {\scriptsize $\vertexf$ } ;
    \node (v7) at (4,-1) [varnode] {\scriptsize $\vertexg$ } ;
    \node (v8) at (5,0) [varnode] {\scriptsize $\vertexh$ } ;
    \node (v9) at (6,-1) [varnode] {\scriptsize $\vertexi$ } ;
    \node (v10) at (6,1) [varnode] {\scriptsize $\vertexj$ } ;
    \node (v11) at (5,3) [varnode] {\scriptsize $\vertexk$ } ;
    \node (v12) at (4,2) [varnode] {\scriptsize $\vertexl$ } ;
    \draw [green,-,matchededge] (v1) -- (v2) ;
    \draw [-,matchededge] (v1) -- (v3) ;
    \draw [-,matchededge] (v2) -- (v3) ;
    \draw [green,-,matchededge] (v3) -- (v4) ;
    \draw [-,matchededge] (v4) -- (v5) ;
    \draw [green,-,matchededge] (v5) -- (v6) ;
    \draw [-,matchededge] (v5) -- (v7) ;
    \draw [-,matchededge] (v6) -- (v7) ;
    \draw [green,-,matchededge] (v7) -- (v8) ;
    \draw [-,matchededge] (v8) -- (v9) ;
    \draw [-,matchededge] (v8) -- (v10) ;
    \draw [green,-,matchededge] (v9) -- (v10) ;
    \draw [-,matchededge] (v10) -- (v11) ;
    \draw [-,matchededge] (v8) -- (v12) ;
  \end{tikzpicture}
  \caption{\label{fig:maxmatch}An undirected graph. The green edges constitute a maximum cardinality matching.}
\end{figure}
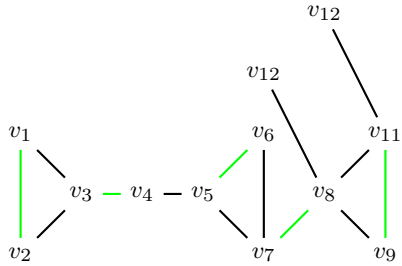

\paragraph{The Algorithm} At a high-level, Edmonds' blossom-shrinking algorithm~\cite{edmond1965blossom} works as follows.
The algorithm has a top-loop that repeatedly searches for an augmenting path, i.e.\ a path whose edges alternate in terms of membership in the matching and that begins and ends at unmatched vertices.
Initially, the current matching is empty.
Whenever an augmenting path is found, augmentation of the matching using the found augmenting path increases the size of the matching by one.
Augmentation is done by taking the symmetric difference between the matching and the edges in the augmenting path, which will always return a larger matching (see Figure~\ref{fig:augpathfig2}).
If no augmenting path exists with respect to the current matching, the current matching has maximum cardinality.

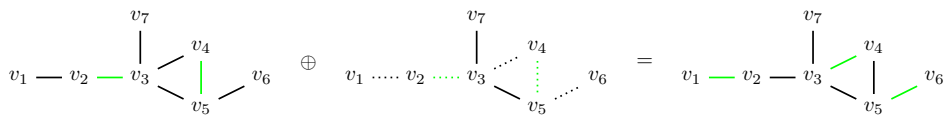
\begin{figure}[H]
  \centering
  \scalebox{0.8}{\begin{tikzpicture}[rotate=0]
    \node (matching) at (0,0)
          {\usebox{1}{\begin{tikzpicture}[rotate=0]
                \node (v1) at (0,0) [varnode] {$\vertexa$} ;
                \node (v2) at (1,0) [varnode] {$\vertexb$} ;
                \node (v3) at (2,0) [varnode] {$\vertexc$} ;
                \node (v4) at (3,0.5) [varnode] {$\vertexd$} ;
                \node (v5) at (3,-0.5) [varnode] {$\vertexe$} ;
                \node (v6) at (4,0) [varnode] {$\vertexf$} ;
                \node (v7) at (2,1) [varnode] {$\vertexg$} ;

                \draw [-,matchededge] (v1) -- (v2) ;
                \draw [green,-,matchededge] (v2) -- (v3) ;
                \draw [-,matchededge] (v3) -- (v4) ;
                \draw [-,matchededge] (v3) -- (v5) ;
                \draw [green,-,matchededge] (v4) -- (v5) ;
                \draw [-,matchededge] (v5) -- (v6) ;
                \draw [-,matchededge] (v3) -- (v7) ;
            \end{tikzpicture}}
          };
          
          \node (symadd) [right= 0.1cm of matching, varnode] {$\oplus$};

          \node (augpath) [right= 0.1cm of symadd]
                {\usebox{1}{\begin{tikzpicture}[rotate=0]
                      \node (v1) at (0,0) [varnode] {$\vertexa$} ;
                      \node (v2) at (1,0) [varnode] {$\vertexb$} ;
                      \node (v3) at (2,0) [varnode] {$\vertexc$} ;
                      \node (v4) at (3,0.5) [varnode] {$\vertexd$} ;
                      \node (v5) at (3,-0.5) [varnode] {$\vertexe$} ;
                      \node (v6) at (4,0) [varnode] {$\vertexf$} ;
                \node (v7) at (2,1) [varnode] {$\vertexg$} ;
                      \draw [-,dotted,matchededge] (v1) -- (v2) ;
                      \draw [green,-,dotted,matchededge] (v2) -- (v3) ;
                      \draw [-,dotted,matchededge] (v3) -- (v4) ;
                      \draw [-,matchededge] (v3) -- (v5) ;
                      \draw [green,-,dotted,matchededge] (v4) -- (v5) ;
                      \draw [-,dotted,matchededge] (v5) -- (v6) ;
                      \draw [-,matchededge] (v3) -- (v7) ;
                  \end{tikzpicture}}
                };

                \node (equal) [right=0.1cm of augpath, varnode] {$=$};

                \node (bigmatching) [right= 0.1cm of equal]
                      {\usebox{1}{\begin{tikzpicture}[rotate=0]
                            \node (v1) at (0,0) [varnode] {$\vertexa$} ;
                            \node (v2) at (1,0) [varnode] {$\vertexb$} ;
                            \node (v3) at (2,0) [varnode] {$\vertexc$} ;
                            \node (v4) at (3,0.5) [varnode] {$\vertexd$} ;
                            \node (v5) at (3,-0.5) [varnode] {$\vertexe$} ;
                            \node (v6) at (4,0) [varnode] {$\vertexf$} ;
                \node (v7) at (2,1) [varnode] {$\vertexg$} ;
                            \draw [green,-,matchededge] (v1) -- (v2) ;
                            \draw [-,matchededge] (v2) -- (v3) ;
                            \draw [green,-,matchededge] (v3) -- (v4) ;
                            \draw [-,matchededge] (v3) -- (v5) ;
                            \draw [-,matchededge] (v4) -- (v5) ;
                            \draw [green,-,matchededge] (v5) -- (v6) ;
                            \draw [-,matchededge] (v3) -- (v7) ;
                      \end{tikzpicture}}};
  \end{tikzpicture}}
  \caption[]{\label{fig:augpathfig2} A figure demonstrating how a matching can be augmented by an agumenting path. The current matching is labelled green. The dotted path is an augmenting path w.r.t\ that matching and the graph on the left. The new matching is on the right, with one more edge.}
\end{figure}

The previous loop is similar to the loops of many other algorithms for matching and, more generally, combinatorial optimisation problems.
The main difficulty, however, for maximum cardinality matching in general graphs is that of searching for augmenting paths.
For instance, one could think of using a modified breadth-first or depth-first search, which takes edges based on alternating membership in the matching, to find those paths.
Although this works for finding augmenting paths in bi-partite graphs, it would not work for general graphs, as the presence of odd-length alternating cycles (henceforth, for brevity, odd cycles) in the graph would necessitate keeping track of whether a given vertex has been visited through an edge in the matching, an edge not in the matching, or both types of edges.
Furthermore, this has to be done for every path, which would make the algorithm run in an exponential worst-case running time.

Edmonds' blossom shrinking algorithm avoids this problem based on the following insight: \emph{problematic odd cycles in the graph can be shrunk (i.e.\ contracted to one vertex) and the resulting graph has an augmenting path iff the original graph has one}.
Such problematic odd cycles occur within so-called \emph{blossoms}.
A blossom is an odd cycle whose base (the vertex on the cycle incident to two non-matching edges) is either unmatched or can be reached from an unmatched vertex by an alternating path ending in a matching edge.
This insight is then used to modify the general schema s.t.\ the search becomes for an augmenting path \emph{or} a blossom.
If the former is found, the matching is augmented.
If the latter is found, the blossom's odd cycle is shrunk into one vertex and the search continues.
Crucially, this works for any blossom that is found, and thus no backtracking is needed, which preserves the polynomial worst-case running time.
Having to prove that shrinking those odd cycles preserves augmenting paths and the fact that the algorithm needs data structures to represent such odd cycles, shrink them, and find them makes Edmonds' blossom shrinking algorithm one of the harder algorithms to justify in the theory of combinatorial optimisation and efficient algorithms more generally.

\paragraph{Methodology} We implement our verification using Isabelle/HOL's \emph{locale}s~\cite{LocalesBallarin}, which provide a mechanism to parametrically model algorithms and to verify them in a step-wise refinement approach.
In this approach, at a given step, we define an algorithm and assume that some functions exist, along with desirable properties of those functions.
Those properties are written as specifications.
In the next step, we define a more detailed description of the assumed function and show that this description satisfies the specification.
This is repeated until there are no significant assumptions left or, if there are any, only trivial ones from an algorithmic perspective are left, e.g.\ the existence of a function that chooses an arbitrary element of a finite set.

\paragraph{Availability} Our formalisation is archived at the DOI \url{10.5281/zenodo.15767314}.
It was done as part of a larger effort to formalise the theory and algorithms for combinatorial optimisation~\footnote{\url{https://github.com/mabdula/Isabelle-Graph-Library}}.

\section{Background}
\label{sec:background}

In this section, we provide a brief overview of Isabelle/HOL and its notation~(\Cref{subsec:isabelle}); our formalisations of undirected graphs~(\Cref{subsec:graphs}), matchings~(\Cref{subsec:matchings}), augmenting paths~(\Cref{subsec:augpaths}); and our approach to modelling non-determinism~(\Cref{subsec:nondet}).

\subsection{Isabelle/HOL Notation}
\label{subsec:isabelle}

Isabelle/HOL~\cite{IsabelleHOLRef} is a theorem prover based on Higher-Order Logic. 
Isabelle's syntax is a variation of Standard ML combined with almost standard mathematical notation.
Function application is written infix and functions are usually Curried (i.e.\ function $f$ applied to arguments $x_1~\ldots~x_n$ is written as $f~x_1~\ldots~x_n$ instead of the standard notation $f(x_1,~\ldots~,x_n)$).
We explain non-standard syntax in the paper where it occurs.

A central concept for our work is Isabelle/HOL \emph{locale}s.
A locale is a named context: definitions and theorems proved within the locale can refer to the parameters and assumptions declared in the heading of the locale.
Here, for instance, we have the locale \isa{graph\_defs} fixing a graph \isa{G} and the locale \isa{graph\_abs} that additionally assumes that the graph satisfies an invariant \isa{graph\_invar}.
We extensively use locales, as we show in the rest of this paper, to structure the reasoning in our formalisation.


\begin{figure}[H]
\begin{lstlisting}[
 language=Isabelle,
 caption={Basic graph definitions: an undirected graph is formalised as a set of sets. We define a context \isa{graph-abs}, where we fix a graph \isa{G} and assume that it has the right type and every edge in it is a doubleton set and that it has a finite set of vertices.},
 label={isabelle:graphDefs},
 captionpos=b
 ]
definition degree where
  "degree G v ≡ card' ({e. v ∈ e} ∩ G)"

locale graph_defs =
  fixes G :: "'a set set"

definition "dblton_graph G ≡ (∀e∈G. ∃u v. e = {u, v} ∧ u ≠ v)"

abbreviation "graph_invar G ≡ dblton_graph G ∧ finite (Vs G)"

locale graph_abs =
  graph_defs +
  assumes graph: "graph_invar G"
\end{lstlisting}
\end{figure}
\newcommand{\degree}{\ensuremath{d}}

\subsection{Graphs}
\label{subsec:graphs}

An \emph{edge} is a set of vertices with size 2.
A \emph{graph} $\graph$ is a finite set of edges.
The \emph{degree} of a vertex $\lvertexgen$ in a graph $\graph$, denoted by $\degree(\graph,\lvertexgen)$, is $\card{\{\edge\mid \edge\in \graph\wedge\lvertexgen\in\edge\}}$.
Although there is a number of other formalisations of undirected graphs~\cite{noschinskiGraphLib,edmondsHypergraphs}, we started our own formalisation.
We formalised the notion of an undirected graph as shown in Listing~\ref{isabelle:graphDefs}.
The main reason we pursued this formalisation is its simplicity: we do not keep track of an explicit set of vertices for each graph.
Instead, the graph's vertices are the union of the vertices of a graph's edges, $\bigcup\graph$, denoted by $\vertices(\graph)$ (\isa{Vs G}, in Isabelle).
This representation has the advantage that one does not have to prove, with every change to the graph, that all graph's edges are incident on the new graph's associated set of vertices.
Nonetheless, it has the disadvantage that it does not allow for graphs with isolated vertices.
However, this constraint has not caused us any problems in this or in the many other projects in which we and others have used it~\cite{RankingIsabelle,ScalingIsabelle,graphLibRepo}.

\begin{figure}[H]
\begin{lstlisting}[
 language=Isabelle,
 caption={A vertex path in an undirected graph is defined as an inductive predicate. We define based on that the notion of a walk between two vertices and the notion of reachability between two vertices.},
 label={isabelle:pathDefs},
 captionpos=b
 ]
context fixes G :: "'a set set" begin
inductive path where
  path0: "path []" |
  path1: "v ∈ Vs G ⟹ path [v]" |
  path2: "{v,v'} ∈ G ⟹ path (v'#vs) ⟹ path (v#v'#vs)"

definition "walk_betw G u p v
              ≡ (p ≠ [] ∧ path G p ∧ hd p = u ∧ last p = v)"

definition reachable where
  "reachable G u v = (∃p. walk_betw G u p v)"
\end{lstlisting}
\end{figure}

A list of vertices $\vertexa\vertexb\dots\vertexgen_n$ is a \emph{path} w.r.t.\ a graph $\graph$ iff every $\{\vertexgen_i, \vertexgen_{i+1}\}\in\graph$.
A path $\vertexa\vertexb\dots\vertexgen_n$ is a \emph{simple path} iff $\vertexgen_i \neq \vertexgen_j$, for every $1\leq i\neq j\leq n$.
We will denote the list of edges $\{\vertexa,\vertexb\}\{\vertexa,\vertexb\}\{\vertexa,\vertexb\}\dots\{\vertexgen_{n-1},\vertexgen_n\}$ occurring in a path $\vertexa\vertexb\dots\vertexgen_n$ by $\edges(\vertexa\vertexb\dots\vertexgen_n)$.
Paths are formally defined recursively in a straightforward fashion as shown in Listing~\ref{isabelle:pathDefs}.
Simple paths in Isabelle are denoted by \isa{distinct}, indicating that their vertices are pairwise distinct.
A path $\vertexa\vertexb\dots\vertexgen_n$ is called a \emph{cycle} if $3 < n$ and $\vertexgen_n = \vertexa$, and we call it an \emph{odd cycle} if $n$ is even.
Note: in informal statements and proofs, we will overload set operations to lists in the obvious fashion.

\begin{figure}[H]
\begin{lstlisting}[
 language=Isabelle,
 caption={Formal definitions of two different notions of connected components. The \isa{connected\_component} is a connected component of vertices and \isa{component\_edges} is a connected component of edges.},
 label={isabelle:conCompDefs},
 captionpos=b
 ]
definition connected_component where
  "connected_component G v = {v'. v' = v \<or> reachable G v v'}"

definition connected_components where
  "connected_components G =
     {vs. ∃v. vs = connected_component G v ∧ v ∈ (Vs G)}"

definition component_edges where
  "component_edges G C = 
     {{x, y} | x y.  {x, y} ⊆ C ∧ {x, y} ∈ G}"

definition components_edges where
  "components_edges G =
     {component_edges G C| C. C ∈ connected_components G}"
\end{lstlisting}
\end{figure}
\newcommand{\comp}{\ensuremath{K}}
\newcommand{\conComp}{\ensuremath{\mathcal{K}}}
\newcommand{\conComps}{\ensuremath{\mathfrak{K}}}
\newcommand{\compEdges}{\ensuremath{\mathcal{E}}}
The \emph{connected component} of a vertex $\vertexgen$, denoted by $\conComp(\graph,\vertexgen)$ is the set of vertices reachable from $\vertexgen$ in the graph $\graph$.
The \emph{connected components} of a graph, denoted by $\conComps(\graph)$, is $ \{\conComp(\graph,\vertexgen)\mid\vertexgen\in\vertices(\graph)\}$.
We define a second notion, the \emph{component edges}, which, for a set of vertices $\vertices$, denoted by $\compEdges(\vertices)$, is the set of edges incident on two vertices in $\vertices$.
In Isabelle, these definitions are formalised as shown in Listing~\ref{isabelle:conCompDefs}.
The connection between the two notions, which is important for proofs that require the two perspectives on connected components, is characterised as follows.
\begin{myprop}
\label{prop:compEdgesDisj}
For any graph $\graph$ and any two vertices $\vertexa$ and $\vertexb$, we have that
$\compEdges(\conComp(\graph,\vertexa)) = \compEdges(\conComp(\graph,\vertexb))$, if $\vertexa = \vertexb$, and $\compEdges(\conComp(\graph,\vertexa)) \cap \compEdges(\conComp(\graph,\vertexb)) = \emptyset$, otherwise.
\end{myprop}
\noindent Another important property of connected components is the following.
\begin{myprop}
\label{prop:inInsertCompE}
For a graph $\graph$ and an edge $\edge$ ($=\{\lvertexgen, \rvertexgen\}$), if $\comp\in\conComps(\{\edge\}\cup\graph)$, then one of the following holds:
\begin{enumerate}
  \item $\comp\in\conComps(\graph)$
  \item $\lvertexgen\not\in\vertices(\graph)$, $\rvertexgen\not\in\vertices(\graph)$, and $\comp = \{\lvertexgen,\rvertexgen\}$,
  \item $\lvertexgen\in\vertices(\graph)$ and $\rvertexgen\not\in\vertices(\graph)$ and $\comp\in(\{\rvertexgen\}\cup\conComp(\graph,\lvertexgen)) \cup (\conComps(\graph)\setminus\conComp(\graph,\lvertexgen))$, or
  \item $\lvertexgen\in\vertices(\graph)$, $\rvertexgen\in\vertices(\graph)$, $\conComp(\graph,\lvertexgen)\neq\conComp(\graph,\rvertexgen)$, and $\comp\in\{\conComp(\graph,\lvertexgen)\cup\conComp(\graph,\rvertexgen)\} \cup (\conComps(\graph)\setminus\conComp(\graph,\lvertexgen)\setminus\conComp(\graph,\rvertexgen))$.
\end{enumerate}
\end{myprop}
This case analysis is crucial for proving facts about connected components by induction on the graph, and it is generally reusable.
The following property of connected components is necessary for proving Berge's Lemma.
\begin{mylem}
\label{lem:conCompPath}
If $\comp\in\conComps(\graph)$ and, for every $\vertexgen\in\vertices(\graph)$, $\degree(\graph,\vertexgen) \leq 2$, then there is a simple path $\path$ s.t.\ $\path$ has exactly all elements of $\comp$.
\end{mylem}
\begin{proof}[Proof sketch]
The proof is by induction on $\graph$.
Let all the variable names in the induction hypothesis (I.H.) be barred, e.g.\ the connected component is $\ihvar{\comp}$.
The base case has an empty graph and is straightforward.
For the step case, we have as an assumption that $\comp\in\conComps(\{\edge\}\cup\graph)$, for some $\edge$, where there are two vertices s.t.\ $\edge=\{\lvertexgen,\rvertexgen\}$.
From Proposition~\ref{prop:compEdgesDisj}, we have to consider the following four cases.
\begin{mycase}
In this case, we can immediately apply the I.H.\ with $\ihvar{\comp}$ assigned to $\comp$, and obtain $\ihvar{\path}$ that is a simple path w.r.t.\ $\graph$ and that has all the vertices of $\comp$.
Since $\graph\subseteq\{\edge\}\cup\graph$, $\ihvar{\path}$ is the required witness.\qed
\end{mycase}
\begin{mycase} In this case, the required path is $\lvertexgen\rvertexgen$.
\qed
\end{mycase}
\begin{mycase}
\label{case:vinUnin}
First, we apply the I.H.\ to $\graph$, where $\ihvar{\comp}$ is instantiated with $\conComp(\graph,\lvertexgen)$. We obtain $\ihvar{\path}$ that is a simple path w.r.t.\ $\graph$.
From the premises of the induction, we know that $\degree((\{\edge\}\cup\graph),\rvertexgen) \leq 2$.
That means that $\degree(\graph,\lvertexgen) \leq 1$, which means that $\ihvar{\path}$ starts or ends with $\lvertexgen$.
Thus $\rvertexgen$ can be appended to either end of $\ihvar{\path}$ (the end at which $\lvertexgen$ is located) and the resulting list of vertices is the required witness.
\qed
\end{mycase}
\begin{mycase}
In this case, we apply the I.H.\ twice, once to $\conComp(\graph,\lvertexgen)$ and $\graph$ and another to $\conComp(\graph,\rvertexgen)$ and $\graph$.
We obtain two paths $\ihvar{\path}_\lvertexgen$ and $\ihvar{\path}_\rvertexgen$, where both are simple paths w.r.t.\ $\graph$ and where the first has the vertices of $\conComp(\graph,\lvertexgen)$ and the second has those of $\conComp(\graph,\rvertexgen)$.
Also each of the two paths has all the vertices of the corresponding connected component.
Following a similar argument to Case~\ref{case:vinUnin}, we have that $\lvertexgen$ is either at beginning or at the end of $\ihvar{\path}_\lvertexgen$, and the same is true for $\rvertexgen$ and $\ihvar{\path}_\rvertexgen$.
The required witness path is $\ihvar{\path}_\lvertexgen\cat\ihvar{\path}_\rvertexgen$, s.t.\ $\lvertexgen$ and $\rvertexgen$ are adjacent.
\qed
\end{mycase}
\end{proof}
\begin{myremark}
\label{remark:symmetry}
The proof of the above lemma is an example of proofs about connected components that use the case analysis implied by Proposition~\ref{prop:inInsertCompE}.
It also has a theme that recurs often in the context of this present formalisation and other formalisations concerned with graph algorithms, namely, the occurrence of symmetries in proofs.
Here, for instance \begin{enumerate*}\item the third case has two symmetric cases, namely, whether $\lvertexgen$ occurs at the beginning or at the end of $\ihvar{\path}$, and \item the fourth case has four symmetric cases, regarding whether $\lvertexgen$ and $\rvertexgen$ occur in the beginning or the end of $\ihvar{\path}_\lvertexgen$ and $\ihvar{\path}_\lvertexgen$, respectively.\end{enumerate*}
In the formal setting, despite spending effort on devising lemmas capturing them, these symmetries caused the formal proofs to be much longer than the informal (e.g.\ both Proposition~\ref{prop:inInsertCompE} and Lemma~\ref{lem:conCompPath} required two hundred lines of formal proof scripts each).
\end{myremark}

\begin{figure}[H]
\begin{lstlisting}[
 language=Isabelle,
 caption={Formal definition of matchings.},
 label={isabelle:matching},
 captionpos=b
 ]
definition matching where
  "matching M \<longleftrightarrow> 
     (∀e1 ∈ M. ∀e2 ∈ M. e1 ≠ e2 ⟹ e1 ∩ e2 = {})"
\end{lstlisting}
\end{figure}

\paragraph{Matchings}
\label{subsec:matchings}

A set of edges $\matching$ is a \emph{matching} iff $\forall e,e'\in\matching.\; e \cap e' = \emptyset$.
In Isabelle/HOL that is modelled as shown in Listing~\ref{isabelle:matching}.
In almost all relevant cases, a matching is a subset of a graph, in which case we call it a matching w.r.t.\ that graph.
Given a matching $\matching$, we say vertex $\vertexgen$ is (un)matched $\matching$ iff $\vertexgen\in\vertices(\matching)$ (does not hold).
For a graph $\graph$, $\matching$ is a maximum matching w.r.t.\ $\graph$ iff for any matching $\matching'$ w.r.t.\ $\graph$, we have that $\cardinality{\matching'} \leq  \cardinality{\matching}$.

\begin{figure}[H]
\begin{lstlisting}[
 language=Isabelle,
 caption={Formal definition of alternating paths.},
 label={isabelle:altList},
 captionpos=b
 ]
inductive alt_list where
"alt_list P1 P2 []" |
"P1 x ⟹ alt_list P2 P1 l ⟹ alt_list P1 P2 (x#l)"

definition matching_augmenting_path where
  "matching_augmenting_path M p ≡ 
    (length p \<ge> 2) ∧
     alt_list (λe. e ∉ M) (λe. e ∈ M) (edges_of_path p) ∧ 
     hd p ∉ Vs M ∧ last p ∉ Vs M"

abbreviation "graph_augmenting_path E M p ≡
  path E p ∧ distinct p ∧ matching_augmenting_path M p"
\end{lstlisting}
\end{figure}

\subsection{Augmenting Paths}
\label{subsec:augpaths}

A list of vertices $\vertexa\vertexb\dots\vertexgen_n$ is an \emph{alternating path} w.r.t.\ a set of edges $\edges$ iff for some $\edges'$ we have that \begin{enumerate*}\item $\edges' = \edges$ or $\edges' = \complement{E}$, \item $\{\vertexgen_i,\vertexgen_{i+1}\}\in\edges'$ holds for all even numbers $i$, where $1 \leq i < n$, and \item $\{\vertexgen_i,\vertexgen_{i+1}\}\not\in\edges'$ holds for all odd numbers $i$, where $1 \leq i \leq n$.\end{enumerate*}
We call a list of vertices $\vertexa\vertexb\dots\vertexgen_n$ an \emph{augmenting path} w.r.t.\ a matching $\matching$ iff $\vertexa\vertexb\dots\vertexgen_n$ is an alternating path w.r.t.\ $\matching$ and $\vertexa,\vertexgen_n\not\in\vertices(\matching)$. 
We call $\gamma$ an augmenting path w.r.t.\ to the pair $\langle\graph,\matching\rangle$ iff it is an augmenting path w.r.t.\ to a matching $\matching$ and is also a simple path w.r.t.\ a graph $\graph$.
Also, for two sets $s$ and $t$, $s \oplus t$ denotes the symmetric difference of the two sets.
We overload $\oplus$ to arguments which are lists in the obvious fashion.

\begin{figure}[H]
\begin{lstlisting}[
 language=Isabelle,
 caption={Basic principles of reasoning about alternating lists.},
 label={isabelle:altListLemmas},
 captionpos=b
 ]
lemma induct_alt_list012:
  assumes "alt_list P1 P2 l"
  assumes "T []"
  assumes "∧x. P1 x ⟹ T [x]"
  assumes "∧x y zs. P1 x ⟹ P2 y ⟹ T zs ⟹ T (x#y#zs)"
  shows "T l"

lemma alternating_length_balanced:
  assumes "alt_list P1 P2 l" "∀x∈set l. P1 x \<longleftrightarrow> \<not> P2 x"
  shows "length (filter P1 l) = length (filter P2 l) \<or>
         length (filter P1 l) = length (filter P2 l) + 1"

lemma alternating_eq_iff_even:
  assumes "alt_list P1 P2 l" "∀x∈set l. P1 x \<longleftrightarrow> \<not> P2 x"
  shows
    "length (filter P1 l) = length (filter P2 l) \<longleftrightarrow>
         even (length l)"

lemma alternating_eq_iff_odd:
  assumes "alt_list P1 P2 l" "∀x∈set l. P1 x \<longleftrightarrow> \<not> P2 x"
  shows
    "length (filter P1 l) = length (filter P2 l) + 1 \<longleftrightarrow> 
        odd (length l)"
\end{lstlisting}
\end{figure}

Alternating paths are formalised as shown Listing~\ref{isabelle:altList}.
We first define an inductive predicate characterising what it means for a list to alternate w.r.t.\ two predicates, and based on that define augmenting paths.
Since this is the first formalisation of a substantial result involving alternating lists, it is worthwhile to display the basic reasoning principles needed to reason about alternating lists.
Those reasoning principles are the induction principle \isa{induct\_alt\_list012} and the other three lemmas in Listing~\ref{isabelle:altListLemmas} that relate the length of an alternating list to the predicate that holds for the last element of the alternating list.
Those four lemmas are the only ones we needed to derive directly from the definition of alternating lists.
We were able to derive all other facts---which we need in the context of of Edmonds' blossom algorithm as well other matching algorithms~\cite{RankingIsabelle}---pertaining to alternating lists from those four lemmas.


\begin{figure}[H]
\begin{lstlisting}[
 language=Isabelle,
 caption={Two non-deterministic functions that we assume: one to create new vertices and the other to choose vertices.},
 label={isabelle:nonDeterminism},
 captionpos=b
 ]
locale create_vert = 
  fixes create_vert::"'a set ⇒ 'a"
  assumes 
    create_vert_works: "finite vs ⟹ create_vert vs ∉ vs"

locale choose = 
  fixes sel
  assumes sel: "⟦finite s; s ≠ {}⟧ ⟹ (sel s) ∈ s"
\end{lstlisting}
\end{figure}

\begin{figure}[H]
\begin{lstlisting}[
 language=Isabelle,
 caption={A function to non-deterministically choose an edge utilising the non-deterministic function for choosing vertices.},
 label={isabelle:chooseEdge},
 captionpos=b
 ]
definition
  "sel_edge G =( 
     let v1 = sel (Vs G);
         v2 = sel (neighbourhood G v1)
     in
        {v1,v2})"

lemma sel_edge: 
  assumes "graph_invar G" "G ≠ {}"
  shows "sel_edge G ∈ G"
\end{lstlisting}
\end{figure}

\subsection{Nondeterminism}
\label{subsec:nondet}

Edmonds' blossom shrinking algorithm has multiple computational steps that are most naturally modelled nondeterministically.
In a locale-based approach nondeterminism could be modelled by assuming functions that perform nondeterministic computation steps.
If the algorithm is to be executed, those functions are to be instantiated with executable implementations.
This includes functions to non-deterministically choose edges from a graph or a matching, a neighbour of a vertex, etc.
In our approach to model that nondeterminism, however, we aspire to limit nondeterminism to a minimum, i.e.\ we wanted to assume a minimal number of nondeterministic functions in the verified algorithm.
The only places where there is nondeterminism in our model of the algorithm are in the locales \isa{choose} and \isa{create\_vert} shown in Listing~\ref{isabelle:nonDeterminism}.
There, we assume the presence of a function that can choose an arbitrary vertex from a finite set of vertices and a function that, given a finite set of vertices, creates a vertex that is not in the give set.
Based on those functions, we define all other nondeterministic functions used throughout the algorithm.
Listing~\ref{isabelle:chooseEdge}, for instance, shows how we define a function that chooses an arbitrary edge from a set of edges.
Concentrating all nondeterminism that way aims at making the generation of an executable implementation a more straightforward process, requiring minimal modifications to our current formalisation.

\section{The Top-Loop}
\label{sec:toploop}

\begin{figure}[H]
\begin{lstlisting}[
 language=Isabelle,
 caption={An Isabelle/HOL locale showing the parameters on which the top-loop of Edmonds' blossom algorithm is parameterised. Most notably, it assumes the presence of a function that returns an augmenting path, if one exists. Note: in Isabelle/HOL, {$[\![P; Q]\!] \Longrightarrow R$} is shorthand for {$P \Longrightarrow Q \Longrightarrow R$}.},
 label={isabelle:findMaxMatch},
 captionpos=b
 ]
locale find_max_match = graph_abs G for G +
  fixes
    aug_path_search::
      "'a set set ⇒ 'a set set ⇒ ('a list) option" 
  assumes
    aug_path_search_complete: 
    "⟦matching M; M ⊆ G; finite M;
      (∃p. graph_augmenting_path G M p)⟧ ⟹ 
        (∃p. aug_path_search G M = Some p)" and
    aug_path_search_sound:
    "⟦matching M; M ⊆ G; finite M; 
      aug_path_search G M = Some p⟧ ⟹
        graph_augmenting_path G M p"
\end{lstlisting}
\end{figure}
\begin{figure}[H]
\begin{lstlisting}[
 language=Isabelle,
 caption={A recursive function modelling the top-loop of Edmonds' blossom shrinking algorithm.},
 label={isabelle:findMaxMatching},
 captionpos=b
 ]
function find_max_matching where
  "find_max_matching M = 
     (case aug_path_search G M of Some p ⇒
        (find_max_matching (M ⊕ (set (edges_of_path p))))
      | _ ⇒ M)"
\end{lstlisting}
\end{figure}

In Isabelle/HOL, Algorithm~\ref{alg:Blossom} is formalised in a parametric fashion within the Isabelle locale \isa{find\_max\_match} whose header is shown in Listing~\ref{isabelle:findMaxMatch}.
The algorithm itself is formalised as the recursive function shown in Listing~\ref{isabelle:findMaxMatching}.
Recall that Algorithm~\ref{alg:Blossom} is parameterised over the function $\AugPathAlg$, which is a function that searches for augmenting paths, if any exists.
To formalise that, we identify \isa{aug\_path\_search} as a parameter of the locale \isa{find\_max\_match}, corresponding to the function $\AugPathAlg$.
The function \isa{aug\_path\_search} should take as input a graph and a matching.
It should return an \isa{('a list) option} typed value, which would either be \isa{Some p}, if a path p is found, or \isa{None}, otherwise.
There is also the function \isa{the} that, given a term of type \isa{'a option}, returns \isa{x}, if the given term is \isa{Some x}, and which is undefined otherwise.

Functions defined within a locale are parameterised by the constants which are fixed in the locale's header.
When such a function is used outside the locale, these parameters must be provided to the function as arguments. So, if \isa{find\_max\_matching} is used outside the locale above, it should take a function which computes augmenting paths as an argument.
Similarly, theorems proven within a locale implicitly have the assumptions of the locale.
So if we use the lemma \isa{find\_max\_matching\_works} (Listing~\ref{isabelle:findMaxMatchingWorks}) outside of the locale, we would have to prove that the functional argument to \isa{find\_max\_matching} satisfies the assumptions of the locale, i.e.\ that argument is a correct procedure for computing augmenting paths.
The use of locales for performing gradual refinement of algorithms allows us to focus on the specific aspects of the algorithm relevant to a refinement stage, with the rest of the algorithm abstracted away.

\paragraph{Correctness}
\label{sec:correct}

The correctness of Algorithm~\ref{alg:Blossom} depends on the assumed behaviour of $\AugPathAlg$, i.e.\ $\AugPathAlg$ has to conform to the following specification in order for Algorithm~\ref{alg:Blossom} to be correct.
\begin{myspec}
\label{spec:AugPath}
$\AugPathAlg(\graph,\matching)$ is an augmenting path w.r.t.\ $\langle\graph,\matching\rangle$, for any graph $\graph$ and matching $\matching$, iff $\graph$ has an augmenting path w.r.t.\ $\langle\graph,\matching\rangle$.
\end{myspec}
\noindent In the formalisation, this specification corresponds to the assumptions on \isa{find\_aug\_path} in the locale \isa{find\_max\_match} shown in Listing~\ref{isabelle:findMaxMatch}.

The correctness proof of Algorithm~\ref{alg:Blossom} is mainly constituted of Berge's lemma~\cite{BergesLemma}, which justifies the total correctness of most maximum cardinality matching algorithms.

\begin{mythm}[Berge 1957~\cite{BergesLemma}]
\label{thm:Berge}
For a graph $\graph$, a matching $\matching$ is maximum w.r.t.\ $\graph$ iff there is no augmenting path $\path$ w.r.t.\ $\langle \graph, \matching\rangle$.
\end{mythm}

\noindent To prove it, we first need the following technical lemma.
\begin{mylem}
\label{lem:symmDiffStruct}
Consider two matchings w.r.t.\ $\graph$, $\matching$, and $\matching'$.
For a connected component $\conComp$ of the graph $\matching\oplus\matching'$, if $\card{\matching'\cap\compEdges(\conComp)} > \card{\matching\cap\compEdges(\conComp)}$, then $\compEdges(\conComp)$ can always be arranged into $\path$,
a simple path w.r.t.\ $\matching\oplus\matching'$.
\end{mylem}
\begin{proof}[Proof sketch]
We perform a case analysis on $\card{\conComp}$.
If $\card{\conComp}\leq 1$, the proof is trivial.
For the case when $\card{\conComp} > 1$, we prove the theorem by contradiction.
We consider the following two cases.
\begin{mycase}[$\path$ is a cycle]
First, from the fact that $\card{\matching'\cap\compEdges(\conComp)} > \card{\matching\cap\compEdges(\conComp)}$ and since $\path$ is an alternating path, we know that $\head(\path)\in\matching'$ and $\last(\path)\in\matching'$, and the first and last edges of $\path$ are both members of $\matching'$.
From the case assumption, we know that $\head(path)=\last(\path)$.
This implies that $2\leq\degree(\matching',\last(\path))$, which is a contradiction since $\matching'$ is a matching.
\end{mycase}
\begin{mycase}[There is $\vertexgen\in{\path\setminus\{\head(\path),\last(\path)\}}$ s.t.\ $3\leq\degree(\matching\oplus\matching',\vertexgen)$] Since $\vertexgen\in\vertices(\matching\oplus\matching')$, then we have that either $2\leq\degree(\matching,\vertexgen)$ or $2\leq\degree(\matching',\vertexgen)$.
In both cases we have a contradiction, since both $\matching$ and $\matching'$ are matchings.
\end{mycase}
\end{proof}

\begin{proof}[Proof sketch of Theorem~\ref{thm:Berge}]
($\Rightarrow$) Assume there is a matching $\matching'$, s.t.\ $\card{\matching'}>\card{\matching}$.
Then there is a connected component $\conComp\in\conComps(\matching\oplus\matching')$ s.t.\ $\card{\matching'\cap\compEdges(\conComp)} > \card{\matching\cap\compEdges(\conComp)}$.
From Lemma~\ref{lem:symmDiffStruct}, we have that all edges in $\conComp$ can be arranged in a path $\path$, s.t.\ the edges in $\path$ alternate in membership of $\matching$ and $\matching'$.
Also, the path will have more edges from $\matching'$ than it does from $\matching$.
That means that $\path$ starts and ends at vertices that are not in $\matching$.
Then $\path$ is an augmenting path w.r.t.\ $\langle\graph,\matching\rangle$.

\noindent($\Leftarrow$) Suppose there is an augmenting path $\path$ w.r.t.\ $\langle\graph,\matching\rangle$.
Then $\path\oplus\matching$ is a matching w.r.t.\ $\graph$ and $\card{\path\oplus\matching}>\card{\matching}$.
We prove that by structural induction on $\path$.
The proof is straightforward for the cases where $\card{\path}\leq 2$.
For $\card{\path} > 2$, we want to show that the theorem holds for a path $\vertexa\vertexb\vertexc\cat\path$.
We apply the induction hypothesis with $\ihvar{\matching}\equiv(\matching\setminus\{\{\vertexb,\vertexc\}\})\cup\{\{\vertexa,\vertexb\}\}$.
Since $\vertexa\vertexb\vertexc\cat\path$ is an alternating path w.r.t.\ $\matching$, then $\vertexc\cat\path$ is also an alternating path w.r.t.\ $\ihvar{\matching}$.
Also recall that $\vertexc\not\in\vertices(\ihvar{\matching})$, by definition of $\ihvar{\matching}$.
Also $\last(\vertexc\cat\path)\not\in\vertices(\ihvar{\matching})$, from the induction premises and by definition of $\ihvar{\matching}$.
Thus, by applying the induction hypothesis, we have that $\ihvar{\matching}\oplus(\vertexc\cat\path)$ is a matching with more edges than $\ihvar{\matching}$.
Note that, since $\ihvar{\matching} = \matching\oplus\{\{\vertexa,\vertexb\},\{\vertexb,\vertexc\}\}$, which follows from the definition of $\ihvar{\matching}$, we have that $\ihvar{\matching}\oplus(\vertexc\cat\path)={\matching}\oplus(\vertexa\vertexb\vertexc\cat\path)$.
This, together with the fact that $\ihvar{\matching}\oplus(\vertexc\cat\path)$ is a matching with more edges than $\ihvar{\matching}$, finish our proof.
\end{proof}

\noindent The functional correctness of Algorithm~\ref{alg:Blossom} is stated in the following corollary.
\begin{mycor}
\label{cor:BlossomWorks}
Assume that $\AugPathAlg(\graph,\matching)$ satisfies Specification~\ref{spec:AugPath}.
Then, for any graph $\graph$, $\BlossomAlg(\graph, \emptyset)$ is a maximum matching w.r.t.\ $\graph$.
\end{mycor}
\begin{proof}[Proof sketch]
The statement follows from Theorem~\ref{thm:Berge} and the fact that $\AugPathAlg(\graph,\matching)$ satisfies Specification~\ref{spec:AugPath}.
\end{proof}
\paragraph{Formalisation}
\noindent The formalised statement of Berge's lemma is shown in Listing~\ref{isabelle:berge}.
\begin{figure}[H]
\begin{lstlisting}[
 language=Isabelle,
 caption={Statement of Berge's lemma in our formalisation.},
 label={isabelle:berge},
 captionpos=b
 ]
theorem Berge:
  assumes matching: "finite M" "matching M" "M ⊆ G"
  shows
   "(∃p. matching_augmenting_path M p ∧ path G p ∧ distinct p) 
         =  (∃M' ⊆ G. matching M' ∧ card M < card M')"
\end{lstlisting}
\end{figure}






\begin{figure}[H]
\begin{lstlisting}[
 language=Isabelle,
 caption={The termination theorem of the top-loop of the algorithm. An implicit assumption here is that \isa{aug\_path\_search} conforms to the assumptions in the locale header.},
 label={isabelle:findMaxMatchingTerminates},
 captionpos=b
 ]
lemma find_max_matching_dom:
  assumes "matching M"" M ⊆ G"" finite M"
  shows "find_max_matching_dom M"
\end{lstlisting}
\end{figure}

\begin{figure}[H]
\begin{lstlisting}[
 language=Isabelle,
 caption={The functional correctness theorem of the top-loop of the algorithm. Similarly to the termination theorem, we assume \isa{aug\_path\_search} conforms to the assumptions in the locale header.},
 label={isabelle:findMaxMatchingWorks},
 captionpos=b
 ]
lemma find_max_matching_works:
  shows "(find_max_matching {}) ⊆ G"
    "matching (find_max_matching {})"
    "∀M. matching M ∧ M ⊆ G ⟹ 
           card M \<le> card (find_max_matching {})"
\end{lstlisting}
\end{figure}

The formalised functional correctness theorem is shown in Listing~\ref{isabelle:findMaxMatchingWorks}.
The theorem has three conclusions: the algorithm returns a subset of the graph, that subset is a matching, and the cardinality of any other matching is bounded by the size of the returned matching.
Note that since it is proved within the locale \isa{find\_max\_match}, it has an implicit assumption that \isa{find\_aug\_path} satisfies the specification \isa{find\_aug\_path\_spec}.
Also note that the algorithm is initialised with the empty matching.

The formal proof of Corollary~\ref{cor:BlossomWorks} is done by computation induction using the induction principle that results from the termination proof of that recursive function.
We use this methodology that is based on Isabelle/HOL's function package~\cite{functionPackageIsabelle} for modelling and reasoning about all the major algorithms that we consider in this paper: we model them as recursive functions and prove facts about them using computation induction.
For such recursive functions, the induction principle as well as the defining equations are conditional on the input being one on which the function is well-defined (e.g.\ inputs for which the predicate \isa{find\_max\_matching\_dom} in Listing~\ref{isabelle:findMaxMatchingTerminates} applies).

The termination proof of the algorithm is based on showing that $\card{\graph\setminus\matching}$ decreases with every recursive call.
The termination theorem is shown in Listing~\ref{isabelle:findMaxMatchingTerminates}, where it is shown that the algorithm terminates, if it starts with a finite matching.
\begin{myremark}
Our proof of Berge's lemma is similar to the exposition in Bondy and Morty's textbook~\cite[Chapter 16]{bondyGraphTheory}, which is a standard textbook on graph theory.
However, there is a significant difference when it comes to the proof Lemma~\ref{lem:symmDiffStruct}.
In the formal proof we have the extra task of showing that the case analysis performed within that lemma is exhaustive.
In all informal treatments, on the other hand, that is not considered.
Informally, it is obvious that if a set of edges cannot be arranged into a path, it is either a cycle or there is a vertex on which two edges are incident.
In the formalisation showing that exhaustiveness took a majority of the effort of proving Lemma~\ref{lem:symmDiffStruct}.
This is an example of a theme which is recurring in formalising reasoning that appeals to graphical or geometric intuition, and was documented by a number of authors, including us~\cite{RankingIsabelle,ScalingIsabelle,GreensIsabelle,hpcpl,poincarebendixson}.
\end{myremark}

\section{Handling Odd Cycles}
\label{sec:Contraction}

\newcommand{\refine}{\textsf{\upshape refine}}
\begin{algorithm}[H]
  \caption{$\AugPathAlg(\graph, \matching)$}\label{alg:findAugPath}
  \label{alg:AugPathAlg}
  \begin{algorithmic}[1]
    \IF{ $\BlossomOrAugPath(\graph, \matching)$ is an augmenting path w.r.t. $\langle\graph,\matching\rangle$}
      \RETURN {$\BlossomOrAugPath(\graph, \matching)$} \label{AugPathAlg:line:augpath}
    \ELSIF{$\BlossomOrAugPath(\graph, \matching)$ is a blossom $\langle \path, \cycle \rangle$ w.r.t.\ $\langle\graph,\matching\rangle$}
      \RETURN {$\refine(\AugPathAlg(\graph/P_{{\complement{\cycle}}}, \matching/P_{\complement{\cycle}}))$} \label{AugPathAlg:line:blossom}
    \ELSE
      \RETURN {no augmenting path found} \label{AugPathAlg:line:none}
    \ENDIF
  \end{algorithmic}
\end{algorithm}

In this step we refine $\AugPathAlg$, which is the function that computes augmenting paths, into a more detailed description.
In our exposition, $\AugPathAlg$, refined as Algorithm~\ref{alg:AugPathAlg}, is a function that handles odd cycles found in the graph by removing them, which is the main insight underlying Edmonds' blossom shrinking algorithm.
It is again parametrically defined, where it depends on the function $\BlossomOrAugPath$.
$\AugPathAlg$ either \begin{enumerate*}[label=(\roman*)]\item returns an augmenting path if $\BlossomOrAugPath$ finds one (Line~\ref{AugPathAlg:line:augpath}), \item  removes (more on that later) an odd cycle from the graph, by contracting it (Line~\ref{AugPathAlg:line:blossom}; notation defined below) and then recursively continues searching for augmenting paths, or \item reports that no augmenting paths exists, if $\BlossomOrAugPath$ finds no odd cycles or augmenting paths (Line~\ref{AugPathAlg:line:none}).
\end{enumerate*}

An important element here is how odd cycles are manipulated.
Odd cycles found by $\BlossomOrAugPath$ are returned in the form of \emph{blossoms}, which is a central concept in Edmonds' algorithm.
A blossom is an alternating path starting with an unmatched vertex that is constituted of two parts: \begin{enumerate*}[label=(\roman*)]\item the \emph{stem}, which is a simple alternating path and, \item an odd cycle, which is a simple (except for the beginning and end) alternating path with an odd number of edges, starting and ending at the same vertex \end{enumerate*}.
For instance, vertices $\vertexa,\vertexb,\vertexc,\vertexd,\vertexe,$ and $\vertexc$ in Figure~\ref{fig:augpathfig2} constitute a blossom.

\begin{figure}[H]
\begin{lstlisting}[
 language=Isabelle,
 caption={The definition of a blossom. Note: \isa{edges\_of\_path} is a function which, given a path, returns the list of edges constituting the path},
 label={isabelle:matchBlossom},
 captionpos=b
 ]
definition odd_cycle where
  "odd_cycle p ≡ 
     (length p \<ge> 3) ∧ odd (length (edges_of_path p)) ∧
      hd p = last p"

definition match_blossom where
  "match_blossom M stem C ≡
      alt_path M (stem @ C) ∧ distinct (stem @ (butlast C)) ∧
      odd_cycle C ∧ hd (stem @ C) ∉ Vs M ∧
      even (length (edges_of_path (stem @ [hd C])))"

abbreviation "blossom G M stem C ≡
   path G (stem @ C) ∧ match_blossom M stem C"
\end{lstlisting}
\end{figure}

In the rest of this section we first introduce further notation; then we prove that blossom contraction preserves augmenting path existence (\Cref{subsec:contractCorrect}); and then we present an algorithm that contracts blossoms and refines augmenting paths found in a contracted graph (\Cref{subsec:contractAlgo}).

\noindent\textit{Further Notation.}
A pair $\langle \vertexa\vertexb\dots\vertexgen_{i-1}, \vertexgen_i\vertexgen_{i+1}\dots\vertexgen_n\rangle$ is a \emph{blossom} w.r.t.\ a matching $\matching$ iff \begin{enumerate*}\item $\vertexgen_i\vertexgen_{i+1}\dots\vertexgen_n$ is an odd cycle, \item $\vertexa\vertexb\dots\vertexgen_n$ is an alternating path w.r.t.\ $\matching$, \item $\vertexa\not\in\vertices(\matching)$, and \item $i$ is an odd number.\end{enumerate*}
We refer to $\vertexa\vertexb\dots\vertexgen_i$ as the \emph{stem} of the blossom and $\vertexgen_i$ as the \emph{base} of the blossom.
In many situations we have a pair {$\langle \vertexa\vertexb\dots\vertexgen_{i-1}, \vertexgen_i\vertexgen_{i+1}\dots\vertexgen_n\rangle$} which is a blossom w.r.t.\ a matching $\matching$ where {$\vertexa\vertexb\dots\vertexgen_{i-1} \vertexgen_i\vertexgen_{i+1}\dots\vertexgen_{n-1}$} is also a simple path w.r.t.\ a graph $\graph$ and {$\{\vertexgen_{n-1},\vertexgen_n\}\in\graph$}.
In this case we call it a blossom w.r.t.\ $\langle\graph,\matching\rangle$.
The formalisation of the notion of a blossom is shown in Listing~\ref{isabelle:matchBlossom}.
Furthermore, for a function $f$ and a set $s$, let $\Img{f}{s}$ denote the image of $f$ on $s$.
For a graph $\graph$, and a function $f$, the \emph{quotient} $\graph/f$ is the set $\{\Img{f}{e}\mid e\in\graph \wedge \cardinality{e} = \cardinality{\Img{f}{e}}\}$.
Let, for a set of vertices $s$, the function $P_s$ be defined as $P_s(\lvertexgen)\equiv \ifnew\; \lvertexgen \in s\; \thennew\; \vertexgen\; \elsenew\; \rvertexgen$, where $\rvertexgen\not\in s$.
In many cases, this function will be considered for the complement of a set, denoted by the superscript $\oldcomplement$.
\todo{Remove quotG if formalisation is in graph abs}

\subsection{Blossom Contraction: A Detailed Proof}
\label{subsec:contractCorrect}

We prove that contracting (i.e. shrinking) the odd cycle of a blossom preserves the existence of an augmenting path, which is the second core result needed to prove the validity of the blossom-shrinking algorithm, after Berge's lemma.
This result is the core idea behind this algorithm and it is why the algorithm can compute maximum cardinality matchings in polynomial time.

\begin{mythm}
\label{thm:quotient}
Consider a graph $\graph$, a matching $\matching$ w.r.t.\ $\graph$, a blossom $\langle \stem, \cycle\rangle$ w.r.t.\ $\langle\graph,\matching\rangle$, and a vertex $\rvertexgen\not\in\vertices(\graph)$.
Then we have an augmenting path w.r.t.\ $\langle\graph,\matching\rangle$ iff there is an augmenting path w.r.t.\ $\langle\graph/P_{\complement{\cycle}},\matching/P_{\complement{\cycle}}\rangle$.
\end{mythm}
\begin{proof}[Proof sketch] We prove the directions of the bi-implication separately.
\begin{myright}
Let $\path$ be the augmenting path w.r.t.\ $\langle\graph,\matching\rangle$.
We prove this direction by considering two main cases.
\begin{mycase}[$\stem=\emptyset$ (i.e.\ $\head(\cycle)\not\in\matching$)]
\label{case:emptyStem}
We have six further cases, which are clearly exhaustive.
\begin{mysubcase}[$\card{\path\cap\cycle}=0$]
\label{case:emptyInter}
This case is trivial, since $\path$ would also be an augmenting path w.r.t.\ $\langle\graph/P_{\complement{\cycle}},\matching/P_{\complement{\cycle}}\rangle$.
\end{mysubcase}
\begin{mysubcase}[$\card{\path\cap\cycle}=1$]
\label{case:singleInter}
In this case, we have three further cases
, each representing a possible position of the odd cycle w.r.t.\ the augmenting path.
\begin{mysubsubcase}[$\head(\cycle)\in\path\setminus\{\head(\path),\last(\path)\}$]
In this case, we have a contradiction because the base of the blossom will be in the matching, contradicting the assumption of Case~\ref{case:emptyStem}.
\end{mysubsubcase}
\begin{mysubsubcase}[$\head(\cycle)=\last(\path)$]
\label{case:baseEnd}
Let $\path'$ be $\path$ without $\last(\path)$, i.e.\ $\path = \path'\cat\last(\path)$.
In this case, the cycle is contracted to $u$.
Also $u\not\in\vertices(\matching/P_{\complement{\cycle}})$. 
Thus $\path'\cat u$ is an augmenting path w.r.t.\ $\langle\graph/P_{\complement{\cycle}},\matching/P_{\complement{\cycle}}\rangle$.
\end{mysubsubcase}
\begin{mysubsubcase}[$\head(\cycle)=\head(\path)$]
\label{case:baseBegin}
This case is symmetric with Case~\ref{case:baseEnd}.
\end{mysubsubcase}
\begin{mysubsubcase}[$\head(\path)\in(\cycle\setminus\{\head(\cycle)\})$]
\label{subsubcase:augPathHeadInCyc}
This leads to a contradiction because $\head(\path)\not\in\matching$, from the assumption that $\path$ is an augmenting path w.r.t.\ $\langle\graph,\matching\rangle$, and every vertex in $\cycle\setminus\head(\cycle)$ is in $\vertices(\matching)$.
\end{mysubsubcase}
\begin{mysubsubcase}[$\last(\path)\in(\cycle\setminus\{\head(\cycle)\})$]
\label{subsubcase:augPathLastInCyc}
This case is symmetric to Case~\ref{subsubcase:augPathHeadInCyc}.
\end{mysubsubcase}
\begin{mysubsubcase}[$(\path\setminus\{\head(\path),\last(\path)\})\cap(\cycle\setminus\head(\cycle)\neq\emptyset$]
\label{subsubcase:augPathMidInCyc}
Note that, from the definition of a blossom, every vertex in the cycle, except for the base, is in $\matching$.
From the assumption of Case~\ref{case:singleInter}, we have that $2 \leq \degree(\matching,\vertexgen)$, for some $\vertexgen\in\cycle\setminus\{\head(\cycle)\}$.
This is a contradiction because $\matching$ is a matching.
\end{mysubsubcase}
\end{mysubcase}
\begin{mysubcase}[$1 < \card{\path\cap\cycle}$]
\label{case:moreInter}
From the case assumption, there must be $\path_1$, $\path_2$, and $\path_3$ s.t.\ $\path=\path_1\cat\path_2\cat\path_3$, $\path_1\cap\cycle=\emptyset$, $\path_3\cap\cycle=\emptyset$, and $\path_2$ is non-empty.
This case can be divided into two further cases.
\begin{mysubsubcase}[$\path_3$ is non-empty and $\{\last(\path_2),\head(\path_3)\}\not\in\matching$]
\label{subsubcase:baseMiddleRight}
In this case, the cycle is contracted to $u$ and $\path_3$ is the same after contraction.
Also $u\not\in\vertices(\matching/P_{\complement{\cycle}})$.
Thus $u\cat\path_3$ is an augmenting path w.r.t.\ $\langle\graph/P_{\complement{\cycle}},\matching/P_{\complement{\cycle}}\rangle$.
\end{mysubsubcase}
\begin{mysubsubcase}[$\path_1$ is non-empty and $\{\last(\path_1),\head(\path_2)\}\not\in\matching$]
\label{subsubcase:baseMiddleLeft}
This case is symmetric with Case~\ref{subsubcase:baseMiddleRight}.
\end{mysubsubcase}
\end{mysubcase}
\end{mycase}
\begin{mycase}[$\stem\neq\emptylist$] From the fact that $\path$ is an augmenting path w.r.t.\ $\langle\graph,\matching\rangle$,
we have that $\matching\oplus\path$ is a matching and $\card{\matching}<\card{\matching\oplus\path}$.
We also have that $\head(\stem)\not\in\matching$ and that $\last(\stem)\in\matching$, from the definition of  blossom.
Thus $\matching\oplus\stem$ is a matching too, but $\card{\matching\oplus\stem}=\card{\matching}$, and thus, from Berge's lemma, there must be $\path'$ that is an augmenting path w.r.t.\ $\langle\graph,\matching\oplus\stem\rangle$.
Accordingly, we can apply Case~\ref{case:emptyStem} to $\path'$ and to the matching $\matching\oplus\stem$.
\end{mycase}
\end{myright}
\begin{myleft}
Let $\path$ be the augmenting path w.r.t.\ $\langle\graph/P_{\complement{\cycle}},\matching/P_{\complement{\cycle}}\rangle$. We have two cases.
\begin{mycase}[$\rvertexgen\not\in\path$] In this case we have that $\path$ is an augmenting path w.r.t.\ $\langle\graph,\matching\rangle$, which finishes our proof.
\end{mycase}
\begin{mycase}[$\rvertexgen\in\path$]
\label{case:uInPath}
From the assumption of Case~\ref{case:uInPath}, there are paths $\path_1$ and $\path_2$, and a vertex $\rvertexgen$, s.t.\ $\path=\path_1\cat \rvertexgen\cat\path_2$.
Since $\path$ is an augmenting path w.r.t.\ $\langle\graph/P_{\complement{\cycle}},\matching/P_{\complement{\cycle}}\rangle$, then exactly one of two edges incident to $\rvertexgen$ belongs to $\matching/P_{\complement{\cycle}}$, which gives rise to the following two cases.
\begin{mysubcase}[$\{\last(\path_1),u\}\in\matching/P_{\complement{\cycle}}$]
\label{case:lastPaMatched}
From the case assumption (namely, $\{\last(\path_1),u\}\in\matching/P_{\complement{\cycle}}$) and from the definition of the quotient operation on graphs, we know that there is some vertex $\vertexa\in\cycle$ s.t.\ $\{\last(\path_1),\vertexa\}\in\matching$. 
Since $\last(\path_1)\not\in\cycle$ and $\{\last(\path_1),\vertexa\}\in\matching$, then $\vertexa=\head(\cycle)$.
We also know that since $\{u,\head(\path_2)\}\in\graph/P_{\complement{\cycle}}$, there must be a vertex $\vertexb$ s.t.\ $\vertexb\in\cycle$ and $\{\vertexb,\head(\path_2)\}\in\graph$.
This means there are $\cycle_1$ and $\cycle_2$ s.t.\ $\cycle=\cycle_1\cat \vertexb\cat\cycle_2$.
We have two cases.
\begin{mysubsubcase}[$\{\last(\cycle_1),\vertexb\}\in\matching$]
\label{case:lastPainQuotrevCb}
In this case, we have that $\path_1\cat\reverse(\cycle_2)\cat\path_2$ is an augmenting path w.r.t.\ $\langle\graph,\matching\rangle$, which finishes our proof.
\end{mysubsubcase}
\begin{mysubsubcase}[$\{\vertexb,\head(\cycle_2)\}\in\matching$]
\label{case:lastPainQuotCa}
In this case, we have that $\path_1\cat\cycle_1\cat\path_2$ is an augmenting path w.r.t.\ $\langle\graph,\matching\rangle$, which finishes our proof.
\end{mysubsubcase}
\end{mysubcase}
\begin{mysubcase}[$\{\head(\path_2),\head(\cycle)\}\in\matching/P_{\complement{\cycle}}$]
\label{case:hdPbMatched}
This is symmetric with case Case~\ref{case:lastPaMatched}. 
\end{mysubcase}
\end{mycase}
\end{myleft}
\end{proof}

\paragraph{Formalisation}

To formalise Theorem~\ref{thm:quotient}, we first declare the locale \isa{quot} shown in Listing~\ref{isabelle:quot}.
This locale makes clear the assumptions under which odd cycles can be contracted: the set of vertices in the odd cycle (\isa{Vs G - s}) has to be non-empty, and the odd cycle is contracted to a vertex that does not occur in the graph outside of the odd cycle.
This allows the odd cycle to be contracted to a representative vertex from the odd cycle, which would allow for one of the more efficient implementations of the algorithm.
The formalisation of our notion of a quotient is shown in Listing~\ref{isabelle:quotG}.
Note that the quotient graph has the looping edge \isa{{u}} removed from it, which could be the image of an edge \isa{e} $\subseteq$ \isa{s}.

\begin{figure}[H]
\begin{lstlisting}[
 language=Isabelle,
 caption={A locale fixing assumptions on the representative vertex \isa{u} for contraction.},
 label={isabelle:quot},
 captionpos=b
 ]
locale pre_quot = choose sel + graph_abs E
  for sel::"'a set ⇒ 'a" and E::"'a set set"

locale quot = pre_quot sel E for sel E +
  fixes s::"'a set" and u::'a
  assumes good_quot_map: "u ∉ s" "s \<subset> Vs E"

abbreviation "P v ≡ (if v ∈ s then v else u)"
\end{lstlisting}
\end{figure}
\todo{$E\rightarrow\graph$}

\begin{figure}[H]
\begin{lstlisting}[
 language=Isabelle,
 caption={The formalisation of a quotient of a graph. Note: \isa{f ` s} is the image of function \isa{f} on a set \isa{s}.},
 label={isabelle:quotG},
 captionpos=b
 ]
definition quot_graph where
  "quot_graph P G = {e'. ∃e∈G. e' = P ` e}"

abbreviation "quotG G ≡ (quot_graph P G) - {{u}}"
\end{lstlisting}
\end{figure}

Now, recall the function $\refine$ that refines a quotient augmenting path to a concrete one.
Its formalisation is shown in Listing~\ref{isabelle:refine}.
The function $\refine$ takes an augmenting path $p$ in the quotient graph and returns it unchanged if it does not contain the vertex $u$ and deletes $u$ and splits $p$ into two paths $p_1$ and $p_2$ otherwise.
In the latter case, $p_1$ and $p_2$ are passed to \isa{replace\_cycle}.
This function first defines two auxiliary paths \isa{stem2p2} and \isa{p12stem} using the function \isa{stem{\isadigit{2}}vert\_path}.
\isa{stem{\isadigit{2}}vert\_path} with last argument \isa{hd p2} uses \isa{choose\_con\_vert} to find a neighbour of \isa{hd p2} on the cycle $C$.
It splits the cycle at this neighbour and then returns the path leading to the base of the blossom starting with a matching edge.
Finally, \isa{replace\_cycle} concatenates together $p_1$, $p_2$ and either \isa{stem2p2} and \isa{p12stem} to obtain an augmenting path in the concrete graph.
This function has many possible execution paths, each equivalent to a case in the analysis in the backward direction of the proof of Theorem~\ref{thm:quotient}.



\begin{figure}[H]
\begin{lstlisting}[
 language=Isabelle,
 caption={The formalisation of $\refine$.},
 label={isabelle:refine},
 captionpos=b
 ]
fun find_pfx::"('b ⇒ bool) ⇒ 'b list ⇒ 'b list" where
 "find_pfx Q [] = []" |
 "find_pfx Q (h # l) = 
   (if (Q h) then [h] else h # (find_pfx Q l))"

definition stem2vert_path where
"stem2vert_path C M v ≡
  let
    find_pfx' = 
      (λC. find_pfx ((=) (choose_con_vert (set C) v)) C)
  in
    if (last (edges_of_path (find_pfx' C)) ∈ M) then
      (find_pfx' C)
    else
      (find_pfx' (rev C))"

definition replace_cycle where
  "replace_cycle C M p1 p2 ≡
   let stem2p2 = stem2vert_path C M (hd p2);
       p12stem = stem2vert_path C M (last p1)
   in
   if p1 = [] then
     stem2p2 @ p2
   else
     (if p2 = [] then 
        p12stem @ (rev p1)
      else
       (if {u, hd p2} ∉ quotG M then
         p1 @ stem2p2 @ p2
       else
         (rev p2) @ p12stem @ (rev p1)))"

definition refine where
  "refine C M p ≡
   if (u ∈ set p) then
     (replace_cycle C M (fst (pref_suf [] u p)) (snd (pref_suf [] u p)))
   else p"
\end{lstlisting}
\end{figure}

The formal statement of Theorem~\ref{thm:quotient} is shown in Listing~\ref{isabelle:quotient}.
Similar to what we mentioned in Remark~\ref{remark:symmetry}, a lot of the effort to formally prove this theorem is dominated by managing symmetries.
One way to handle these symmetries is to devise lemmas capturing them, like lemma \isa{path\_u\_p3} shown in Listing~\ref{isabelle:symmetryBreak}.
That lemma covers a major part of proving the four cases arising in Cases~\ref{case:lastPaMatched}~and~\ref{case:hdPbMatched} in the proof of Theorem~\ref{thm:quotient}.
However, despite devising many of these lemmas, the formal proof of this theorem was around 5K lines of proof script.
This is because, although lemmas like that save a lot of repeated reasoning, it can still be cumbersome to prove that their preconditions hold for all cases (e.g.\ when \isa{C} is assigned with $\cycle_1$, $\cycle_2$, $\reverse(\cycle_1)$, and $\reverse(\cycle_2)$ from Cases~\ref{case:lastPaMatched}~and~\ref{case:hdPbMatched}).

\begin{figure}[H]
\begin{lstlisting}[
 language=Isabelle,
 caption={The formal statements of the two directions of Theorem~\ref{thm:quotient}.},
 label={isabelle:quotient},
 captionpos=b
 ]
theorem refine_works:
  assumes
    cycle: "blossom E M stem C" and
    quot_aug_path: 
      "graph_augmenting_path (quotG E) (quotG M) p'" and 
    matching: "matching M" "M ⊆ E" and
    quot: "s = (Vs E) - set C"
  shows "graph_augmenting_path E M (refine C M p')"

theorem aug_path_works_in_contraction:
  assumes match_blossom: "blossom E M stem C" and
    aug_path: "graph_augmenting_path E M p" and 
    matching: "matching M" "M ⊆ E" "finite M" and
    quot: "s = (Vs E) - set C" "u ∉  Vs E"
  shows "∃p'. graph_augmenting_path (quotG E) (quotG M) p'"
\end{lstlisting}
\end{figure}
\begin{figure}[H]
\begin{lstlisting}[
 language=Isabelle,
 caption={A lemma devised to prove symmetric cases in the proof of \isa{aug\_path\_works\_} \isa{in\_contraction}.},
 label={isabelle:symmetryBreak},
 captionpos=b
 ]
  have path_u_p3: "path (quotG E) (u # p3)"
    if wx: "p = p1 @ x # p3" "x ∈ set C"
           "∀x∈set p3. x ∉ set C" and
      aug_path: "path E p" and
      p3_subset_s: "set p3 ⊆ s" and
      p3_nempty: "p3 ≠ []"
    for p1 x p3 p
\end{lstlisting}
\end{figure}

\subsection{The Algorithm}
\label{subsec:contractAlgo}

Finally, we formalise Algorithm~\ref{alg:findAugPath} parametrically as shown in Listings~\ref{isabelle:findAugPathLocale} and~\ref{isabelle:findAugPath}.
We model the algorithm in Isabelle/HOL using a locale which parameterises the algorithm over the function \isa{blos\_search} which performs the search for blossoms or augmenting paths.
The function either returns an augmenting path or a blossom, which is formalised as the algebraic data type \isa{match\_blossom\_res}.
In addition to being parameterised over the function \isa{blos\_search}, it is also parameterised over the two functions \isa{create\_vert} and \isa{sel}, shown in Listing~\ref{isabelle:nonDeterminism}.
Note that we instantiate both arguments \isa{P} and \isa{s} of the locale \isa{quot} to obtain the quotienting function \isa{quotG} and the function for refining augmenting paths \isa{refine}.

\begin{figure}[H]
\begin{lstlisting}[
 language=Isabelle,
 caption={The locale parameterising the formalisation of the algorithm \isa{find\_aug\_path}. Note: \isa{@} denotes list concatenation.},
 label={isabelle:findAugPathLocale},
 captionpos=b
 ]
datatype 'a match_blossom_res =
  Path "'a list" 
  | Blossom (stem_vs: "'a list") (cycle_vs: "'a list")

locale find_aug_path = choose + create_vert +
  fixes blos_search::"'a set set ⇒ 'a set set ⇒  ('a match_blossom_res) option"
  assumes
    bloss_algo_complete: 
    "⟦graph_invar G; matching M; M ⊆ G;
        (∃p. graph_augmenting_path G M p)⟧
     ⟹ (∃blos_comp. blos_search G M = Some blos_comp)" and
    bloss_algo_sound:
    "⟦graph_invar G; matching M; M ⊆ G; blos_search G M = Some (Path p)⟧ ⟹
     graph_augmenting_path G M p"
    "⟦graph_invar G; matching M; M ⊆ G; blos_search G M = Some (Blossom stem C)⟧ ⟹
     blossom G M stem C"
\end{lstlisting}
\end{figure}

\begin{figure}[H]
\begin{lstlisting}[
 language=Isabelle,
 caption={Formalisation of Algorithm~\ref{alg:findAugPath}.},
 label={isabelle:findAugPath},
 captionpos=b
 ]
function find_aug_path where
  "find_aug_path G M = 
     (case blos_search G M of Some match_blossom_res ⇒
        case match_blossom_res of Path p ⇒ Some p
        | Blossom stem cyc ⇒
            let
              u = create_vert (Vs G);
              s = Vs G - (set cyc);
              quotG = quot.quotG s (create_vert (Vs G));
              refine = 
                quot.refine sel G s (create_vert (Vs G)) cyc M
            in (case find_aug_path (quotG G) (quotG M) of Some p' ⇒ Some (refine p')
                | _ ⇒ None)
      | _ ⇒ None)"
\end{lstlisting}
\end{figure}

\paragraph{Correctness} To prove that Algorithm~\ref{alg:findAugPath} is correct, we first precisely specify the properties expected of $\BlossomOrAugPath$, on which Algorithm~\ref{alg:findAugPath} is parameterised.
\begin{myspec}
\label{spec:BlossomOrAugPath}
For a graph $\graph$ and a matching $\matching$ w.r.t.\ $\graph$, there is a blossom or an augmenting path w.r.t.\ $\langle\graph,\matching\rangle$ iff $\BlossomOrAugPath(\graph, \matching)$ is a blossom or an augmenting path w.r.t.\ $\langle\graph,\matching\rangle$.
\end{myspec}
\noindent The functional correctness Algorithm~\ref{alg:findAugPath} is stated as follows.
\begin{mycor}
\label{cor:AugPathAlgWorks}
Assume $\BlossomOrAugPath(\graph, \matching)$ satisfies Specification~\ref{spec:BlossomOrAugPath}.
Then $\AugPathAlg$ satisfies Specification~\ref{spec:AugPath}.
\end{mycor}
\begin{proof}[Proof sketch] From Theorem~\ref{thm:quotient} and by computation induction.
\end{proof}


\begin{myremark}
Theorem~\ref{thm:quotient} is used in most expositions of the blossom-shrinking algorithm.
In our proof for the forward direction (if an augmenting path exists w.r.t.\ $\langle\graph,\matching\rangle$, then there is an augmenting path w.r.t.\ $\langle\graph/P,\matching/P\rangle$, i.e.\ w.r.t.\ the quotients), we follow a standard textbook approach (e.g.\ Lemma 10.25 in Korte and Vygen's book~\cite{KorteVygenOptimisation}).
Our proof is, nonetheless, the only one we are aware of that explicitly pins down the cases, at least among standard textbooks~\cite{LEDAbook,schrijverBook,KorteVygenOptimisation} and lecture notes available online.
One particular approach that is worth mentioning is that taken in LEDA~\cite{LEDAbook} by Mehlhorn and N\"aher.
In their approach, they skip showing this direction completely, due to the complexity of the case analysis and the fact that it was not fully performed in other expositions.
Instead, they replaced it with a claim, presumed to be much easier to prove, that, if we can construct an odd set cover for $\graph/P$, i.e.\ a certificate can be constructed showing that there is not an augmenting path w.r.t.\ $\langle\graph/P,\matching/P\rangle$, then there is a certificate showing that there is not an augmenting path w.r.t.\ $\langle\graph,\matching\rangle$.
Nonetheless, it turned out, when we tried to formalise that approach that we need a case analysis that is more complex than the one we perform in the proof of Theorem~\ref{thm:quotient}.
\end{myremark}

\begin{myremark}
In our proof for the backward direction (an augmenting path w.r.t. the quotients can be lifted to an augmenting path w.r.t.\ the original graph) we define an (almost) executable function $\refine$ that does the lifting.
We took the choice of explicitly defining that function with using it in an implementation of the algorithm in mind.
This is similar to the approach used in the informal proof of soundness of the variant of the blossom-shrinking algorithm used in LEDA~\cite{LEDAbook}.
\end{myremark}

\section{Computing Blossoms and Augmenting Paths}
\label{sec:BlossomOrAugPath}

\newcommand{\computaltpath}{\textsf{\upshape compute\_alt\_path}}
\renewcommand{\choice}{\textsf{\upshape choose}\;}
\newcommand{\compref}{\textsf{\upshape longest\_disj\_pref}}
\newcommand{\drop}{\textsf{\upshape drop}}
\newcommand{\hd}{\head}
\providecommand{\reverse}{\textsf{\upshape rev}}
\newcommand{\consblossom}{$\langle \reverse(\drop\;(\cardinality{\path_1'}-1)\;\path_1), (\reverse\; \path_1') \cat \path_2'\rangle$}

\begin{algorithm}[t]
    \caption{$\BlossomOrAugPath(\graph, \matching)$}\label{alg:BlossomOrAugPath}
    \begin{algorithmic}[1]
      \IF{$\exists e\in\graph. e \cap \vertices(\matching) = \emptyset$}
        \RETURN Augmenting path $\choice \{e\mid e\in\graph \wedge e \cap \vertices(\matching) = \emptyset\}$
      \ELSIF{$\computaltpath(\graph,\matching) = \langle\path_1,\path_2\rangle$}
        \IF{$\last\;\path_1\neq\last\;\path_2$}
           \RETURN Augmenting path $(\reverse\;\path_1)\cat\path_2$
        \ELSE
           \STATE $\langle\path_1',\path_2'\rangle=\compref(\path_1,\path_2)$ 
           \RETURN Blossom \consblossom
        \ENDIF
      \ELSE
        \RETURN No blossom or augmenting path found
      \ENDIF
    \end{algorithmic}
\end{algorithm}

Here we take one further step in our refinement of the algorithm's description, where we give a more detailed description of the function $\BlossomOrAugPath$ (see Algorithm~\ref{alg:BlossomOrAugPath}), which can compute augmenting paths or blossoms, if any exist in the graph.
The algorithm takes as input two alternating paths and returns either an augmenting path or a blossom.
The two given alternating paths have to satisfy a number of conditions if they are to correspond to an augmenting path or a blossom.
The algorithm is parameterised over two functions: $\computaltpath$ and $\compref$.
The former is the core search procedure of Edmonds' blossom algorithm and the latter is a function that takes the output of the former and uses it to return the stem of a blossom, if the two alternating path returned by $\computaltpath$ represent a blossom.

In Isabelle/HOL, the definition of $\BlossomOrAugPath$ is shown in Listing~\ref{isabelle:BlossomOrAug}.
It depends on the function \isa{longest\_disj\_pfx}, whose definition as well as its correctness statement are in Listing~\ref{isabelle:longestDisjPfx}.
It also depends on the function $\computaltpath$.
We use a locale again to formalise Algorithm~\ref{alg:BlossomOrAugPath}.
That locale parameterises $\BlossomOrAugPath$ on the function $\computaltpath$ that searches for alternating paths and poses the correctness assumptions for that alternating path search function.
The locale assumptions assert that $\computaltpath$ conforms to Specification~\ref{spec:computaltpath}.
\begin{figure}[H]
\begin{lstlisting}[
 language=Isabelle,
 caption={Formalisation of Algorithm~\ref{alg:BlossomOrAugPath}.},
 label={isabelle:BlossomOrAug},
 captionpos=b
 ]
definition compute_match_blossom where
  "compute_match_blossom ≡ 
   (if (∃e. e ∈ unmatched_edges) then
         let singleton_path = sel_unmatched in
           Some (Path singleton_path)
    else
     case compute_alt_path
       of Some (p1,p2) ⇒ 
         (if (set p1 ∩ set p2 = {}) then
            Some (Path ((rev p1) @ p2))
          else
            (let (pfx1, pfx2) = longest_disj_pfx p1 p2 in
              (Some (Blossom
                       (rev (drop (length (the pfx1)) p1))
                       (rev (the pfx1) @ (the pfx2))))))
       | _ ⇒ None)"
\end{lstlisting}
\end{figure}

\begin{figure}[H]
\begin{lstlisting}[
 language=Isabelle,
 caption={An algorithm to find the longest disjoint prefix of two lists and its correctness statement.},
 label={isabelle:longestDisjPfx},
 captionpos=b
 ]
fun longest_disj_pfx where
  "longest_disj_pfx l1 [] = (None,None)"
| "longest_disj_pfx [] l2 = (None,None)"
| "longest_disj_pfx l1 (h#l2) = 
    (let l1_pfx = (find_pfx ((=) h) l1) in
       if (last l1_pfx = h) then
         (Some l1_pfx,Some [h])
       else (let 
              (l1_pfx,l2_pfx) = (longest_disj_pfx l1 l2)
             in
               case l2_pfx of Some pfx2 ⇒ 
                                (l1_pfx,Some (h#pfx2))
                              | _ ⇒ (l1_pfx, l2_pfx)))"

lemma common_pfxs_form_match_blossom':
  assumes
    pfxs_are_pfxs:
      "(Some pfx1, Some pfx2) = longest_disj_pfx p1 p2" and
    from_tree: "p1 = pfx1 @ p" "p2 = pfx2 @ p" and
    alt_paths:
      "alt_path M (hd p2 # p1)" "alt_path M (hd p1 # p2)"
      "last p1 ∉ Vs M" and
    hds_neq: "hd p1 ≠ hd p2" and
    odd_lens: "odd (length p1)" "odd (length p2)" and
    distinct: "distinct p1" "distinct p2" and
    matching: "matching M"
  shows
    "match_blossom M 
                   (rev (drop (length pfx1) p1))
                   (rev pfx1 @ pfx2)"
\end{lstlisting}
\end{figure}

\providecommand{\listgen}{\text{xs}}

\noindent\textit{Further Notation.}
We first introduce some notions and notation.
For a list $\listgen$, let $\cardinality{\listgen}$ be the length of $\listgen$.
For a list $\listgen$ and a natural number $n$, let $\drop\; n\; \listgen$ denote the list $\listgen$, but with the first $n$ elements dropped.
For a list $\listgen$, let $x::\listgen$ denote adding an element $x$ to the front of a list $\listgen$.
For a non-empty list $\listgen$, let $\hd\; \listgen$  and $\last\; \listgen$ denote the first and last elements of $\listgen$, respectively.
Also, for a list $\listgen$, let $\reverse\; \listgen$ denote its reverse.
For two lists $\listgen_1$ and $\listgen_2$, let $\listgen_1\cat \listgen_2$ denote their concatenation.
Also, let $\compref\; \listgen_1\; \listgen_2$ denote the pair of lists $\langle \listgen_1',\listgen_2'\rangle$, s.t.\ $\listgen_1 = \listgen_1'\cat \listgen$ and $\listgen_2 = \listgen_2'\cat \listgen$, for some $\listgen$, where both $\listgen'_1$ and $\listgen'_2$ are disjoint except at their endpoints.
Listing~\ref{isabelle:longestDisjPfx} shows an implementation of the function.
Note: this function is not always well-defined, but it is always well-defined if both lists are paths in a tree ending at the root, which is always the case for its inputs in our context here since they will be traversals from the root of the search tree.

\paragraph{Correctness}
The hard part of reasoning about the correctness of $\BlossomOrAugPath$ is the specification of the properties of the functions on which it is parameterised.
For two paths $\path_1$ and $\path_2$, a graph $\graph$, and a matching $\matching$ consider the following properties: \begin{enumerate}[label=P\arabic*]
\item \label{property:simplePath} $\path_1$ and $\path_2$ are simple paths w.r.t.\ $\graph$.
\item \label{property:altPath} $\path_1$ and $\path_2$ alternating paths w.r.t.\ $\matching$.
\item \label{property:oddLength} $\path_1$ and $\path_2$ are of odd length.
\item \label{property:lastNinMatchingA} $\last\;\path_1\not\in\vertices(\matching)$.
\item \label{property:lastNinMatchingB} $\last\;\path_2\not\in\vertices(\matching)$.
\item \label{property:firstVtxsEdge} $\{\hd\;\path_1,\hd\;\path_2\}\in\graph$.
\item \label{property:firstNinMatching}  $\{\hd\;\path_1,\hd\;\path_2\}\not\in\matching$.\end{enumerate}
These properties are formalised in Listing~\ref{isabelle:computeAltPathSpec}.
\begin{myspec}
\label{spec:computaltpath}
The function $\computaltpath(\graph,\matching)$ returns two lists of vertices $\langle\path_1,\path_2\rangle$ s.t.\ both lists satisfy properties \ref{property:simplePath}-\ref{property:firstNinMatching} iff two lists of vertices satisfying those properties exist.
\end{myspec}
We directly define the second function ($\compref$) and prove it correct rather than devising a specification, mainly due its simplicity, as shown in Listing~\ref{isabelle:longestDisjPfx}.

\begin{mylem}
\label{lem:consAugPath}
Assume $\path_1$ and $\path_2$ satisfy properties \ref{property:simplePath}-\ref{property:firstNinMatching} and are both disjoint, and we have that $\last\;\path_1\neq\last\;\path_2$.
Then $(\reverse\;\path_1)\cat\path_2$ is an augmenting path w.r.t.\ $\langle\graph,\matching\rangle$.
\end{mylem}
\begin{proof}[Proof sketch]
The lemma follows from the following two facts.
\begin{mylem}
$(\reverse\;\path_1)\cat\path_2$ is an alternating path w.r.t.\ $\matching$.
\end{mylem}
\begin{proof}
From \ref{property:lastNinMatchingA}, we have that $\last\;\edges(\path_1)\not\in\matching$, and thus
we have that $\hd\;\edges(\reverse\;\path_1)\not\in\matching$.
Also, from \ref{property:altPath}, \ref{property:oddLength}, and \ref{property:lastNinMatchingB}, we have that $\hd\;\edges(\path_2)\in\matching$.
From that, in addition to \ref{property:altPath} and \ref{property:firstNinMatching}, we finish the proof.
Also, from \ref{property:lastNinMatchingA}, and \ref{property:lastNinMatchingB}, we have that the first and last vertices of $(\reverse\;\path_1)\cat\path_2$ are unmatched.
Accordingly, we have that $(\reverse\;\path_1)\cat\path_2$ is an augmenting path w.r.t.\ $\matching$.
\end{proof}
\begin{mylem}
$(\reverse\;\path_1)\cat\path_2$ is a simple path w.r.t.\ $\graph$.
\end{mylem}
\begin{proof}
This follows from \ref{property:simplePath}, \ref{property:firstVtxsEdge}, and since we assume that $\path_1$ and $\path_2$ are disjoint.
\end{proof}
\end{proof}

\begin{mylem}
\label{lem:consBlos}
If $\path_1$ and $\path_2$ are both \begin{enumerate*} \item simple paths w.r.t.\ $\graph$, \item alternating paths w.r.t.\ $\matching$, and \item of odd length, and if we have that \item $\last\;\path_1=\last\;\path_2$, \item  $\last\;\path_1\not\in\vertices(\matching)$, \item $\{\hd\;\path_1,\hd\;\path_2\}\in\graph$, \item $\{\hd\;\path_1,\hd\;\path_2\}\not\in\matching$, and \item $\langle\path_1',\path_2'\rangle=\compref(\path_1,\path_2)$,\end{enumerate*} then \consblossom is a blossom w.r.t.\ $\langle\graph,\matching\rangle$.
\end{mylem}
\begin{proof}[Proof sketch]
The proof here is done using a similar construction to what we did in the proof of Lemma~\ref{lem:consAugPath}.
\end{proof}

Finally, the following Theorem shows $\BlossomOrAugPath$ is correct.
\begin{mythm}
\label{thm:computaltpath}
Assume that $\computaltpath$ satisfies Specification~\ref{spec:computaltpath}.
Then $\BlossomOrAugPath$ satisfies Specification~\ref{spec:BlossomOrAugPath}.
\end{mythm}
\begin{proof}[Proof sketch]
The theorem follows from Lemma~\ref{lem:consBlos} and Lemma~\ref{lem:consAugPath} and the definitions of Specification~\ref{spec:computaltpath} and $\BlossomOrAugPath$. 
\end{proof}
\begin{figure}[H]
\begin{lstlisting}[
 language=Isabelle,
 caption={The specification of the correctness of the core search procedure.},
 label={isabelle:computeAltPathSpec},
 captionpos=b
 ]
definition compute_alt_path_spec where 
  "compute_alt_path_spec G M compute_alt_path ≡
   (∀p1 p2 pref1 x post1 pref2 post2. 
       compute_alt_path = Some (p1, p2) ⟹
         p1 = pref1 @ x # post1 ∧ p2 = pref2 @ x # post2
          ⟹ post1 = post2) ∧
   (∀p1 p2. compute_alt_path = Some (p1, p2) ⟹
         alt_path M (hd p1 # p2)) ∧
   (∀p1 p2. compute_alt_path = Some (p1, p2) ⟹
        alt_path M (hd p2 # p1)) ∧
   (∀p1 p2. compute_alt_path = Some (p1, p2) ⟹ 
        last p1 ∉ Vs M) ∧
   (∀p1 p2. compute_alt_path = Some (p1, p2) ⟹ 
         last p2 ∉ Vs M) ∧
   (∀p1 p2. compute_alt_path = Some (p1, p2) ⟹
         hd p1 ≠ hd p2) ∧
   (∀p1 p2. compute_alt_path = Some (p1, p2) ⟹
         odd (length p1)) ∧
   (∀p1 p2. compute_alt_path = Some (p1, p2) ⟹
         odd (length p2)) ∧
   (∀p1 p2. compute_alt_path = Some (p1, p2) ⟹
         distinct p1) ∧
   (∀p1 p2.
       compute_alt_path = Some (p1, p2) ⟹ distinct p2) ∧
   (∀p1 p2. compute_alt_path = Some (p1, p2) ⟹ path G p1) ∧
   (∀p1 p2. compute_alt_path =  Some(p1, p2) ⟹ path G p2) ∧
   (∀p1 p2. 
      compute_alt_path = Some (p1, p2) ⟹ {hd p1, hd p2} ∈ G)"

locale match = graph_abs G for G + 
  fixes M
  assumes matching: "matching M" "M ⊆ G"

locale compute_match_blossom' = match G M + choose sel
   for sel::"'a set ⇒ 'a" and G M ::"'a set set" +

fixes compute_alt_path:: "(('a list × 'a list) option)"
assumes 
compute_alt_path_spec:
"compute_alt_path_spec G M compute_alt_path" and
compute_alt_path_complete:
"(((∃p. path G p ∧ distinct p ∧ 
        matching_augmenting_path M p)))
         ⟹ (∃blos_comp. compute_alt_path = Some blos_comp)"
\end{lstlisting}
\end{figure}

\begin{myremark}
The formal proofs of the above lemmas are largely straightforward.
The main difficulty is coming up with the precise properties, e.g.\ in Specification~\ref{spec:computaltpath}, which required many iterations between the correctness proof of Algorithm~\ref{alg:BlossomOrAugPath} and Algorithm~\ref{alg:compAltPath}, which implements the assumed function $\computaltpath$.
\end{myremark}

\section{Searching for an Augmenting Path or a Blossom}
\label{sec:compAltPath}

\newcommand{\lab}{\textsf{\upshape label}}
\newcommand{\even}{\textsf{\upshape even}}
\newcommand{\odd}{\textsf{\upshape odd}}
\newcommand{\examined}{\textsf{\upshape ex}}
\newcommand{\unexamined}{\textsf{\upshape unex\_with\_even}}
\newcommand{\parent}{\textsf{\upshape parent}}
\newcommand{\follow}[1]{\textsf{\upshape follow}\ifstrempty{#1}{}{\;#1\;}}
\providecommand{\comment}[1]{{\scriptsize // #1}}

\begin{algorithm}[t]
    \caption{$\computaltpath(\graph,\matching)$}\label{alg:compAltPath}
    \begin{algorithmic}[1]
      \STATE $\examined:=\emptyset$ \comment{Set of examined edges}
      \FOR{$\vertexgen\in \vertices(\graph)$}
        \STATE $\lab\;\vertexgen := $ None
        \STATE $\parent\;\vertexgen := $ None
      \ENDFOR
      \FOR{$\vertexgen\in \vertices(\graph) \setminus\vertices(\matching)$}\label{compAltPath:labelUnmatched}
        \STATE $\lab\;\vertexgen := \langle \vertexgen,\even\rangle$
      \ENDFOR
      \WHILE{$(\graph\setminus \examined) \cap \{e\mid\exists \vertexgen\in e, r\in\vertices(\graph). \lab\; \vertexgen = \langle r, \even\rangle\}\neq\emptyset$}\label{compAltPath:whileCond}
        \STATE \comment{Choose a new edge and label it examined}\label{compAltPath:iterBegin}
        \STATE $\{\vertexa,\vertexb\}:=\choice (\graph\setminus \examined) \cap \{\{\vertexa,\vertexb\}\mid\exists r. \lab\; \vertexa = \langle r, \even\rangle\}$\label{compAltPath:chooseVAVB}
        \STATE $\examined:=\examined \cup \{\{\vertexa,\vertexb\}\}$
        \IF{$\lab\; \vertexb = $ None}\label{compAltPath:ifCondA}
          \STATE \comment{Grow the discovered set of edges from $r$ by two}
          \STATE $\vertexc:=\choice \{\vertexc \mid \{\vertexb, \vertexc\} \in \matching\}$\label{compAltPath:chooseVC}
          \STATE $\examined:=\examined \cup \{\{\vertexb,\vertexc\}\}$\label{compAltPath:addex}
          \STATE $\lab\; \vertexb := \langle r, \odd \rangle$\label{compAltPath:labVB}
          \STATE $\lab\; \vertexc := \langle r, \even\rangle$\label{compAltPath:labVC}
          \STATE $\parent\;\vertexb := \vertexa$\label{compAltPath:parVB}
          \STATE $\parent\;\vertexc := \vertexb$\label{compAltPath:ifCondAEnd}
        \ELSIF{$\exists s\in\vertices(\graph). \lab\; \vertexb = \langle s, \even\rangle$} \label{compAltPath:elseCondB}
          \STATE \comment{Return two paths from current edge's tips to unmatched vertex(es)}
          \RETURN $\langle\follow{\parent} \vertexa, \follow{\parent} \vertexb\rangle$\label{compAltPath:retA}
        \ENDIF
      \ENDWHILE
      \RETURN No paths found \label{compAltPath:retB}
    \end{algorithmic}
\end{algorithm}

Lastly, we refine the function $\computaltpath$ to a detailed algorithmic description (see Algorithm~\ref{alg:compAltPath}).
This algorithm performs an \emph{alternating tree search}.
The search aims to either find an augmenting path or a blossom.
It is done via growing alternating trees rooted at unmatched vertices.
The search is initialised by making each unmatched vertex a root of an alternating tree; the matched nodes are in no tree initially.
In an alternating tree, vertices at an even depth are entered by a matching edge, vertices at an odd depth are entered by a non-matching edge, and all leaves have even depth. 
In each step of the search, one considers a vertex $\vertexa$ of even depth that is incident to an edge $\{\vertexa,\vertexb\}$ that was not examined yet, s.t.\ there is $\{\vertexb,\vertexc\}\in\matching$.
If $\vertexb$ is not in a tree yet, then one adds $\vertexb$ (at an odd level) and $\vertexc$ (at an even level).
If $\vertexb$ is already in a tree and has an odd level then one does nothing as one simply has discovered another odd length path to $\vertexb$.
If $\vertexb$ is already in a tree and has an even level then one has either discovered an augmenting path (if $\vertexa$ and $\vertexb$ belong to trees with different roots) or a blossom (if $\vertexa$ and $\vertexb$ belong to the same tree).
If the function positively terminates, i.e.\ finds two vertices with even labels, it returns two alternating paths by ascending the two alternating trees to which the two vertices belong, where both paths satisfy Properties~\ref{property:simplePath}-\ref{property:firstNinMatching}.
This tree ascent is performed by the function $\follow{}$.
That function takes a higher-order argument and a vertex.
The higher-order argument is a function that maps every vertex to another vertex, which is intended to be its parent in a tree structure.
\todo{Is there a name for a function that keeps following pointers}

To formalise Algorithm~\ref{alg:compAltPath} in Isabelle/HOL, we first formally define the function $\follow{}$, which follows a vertex's parent, as shown in Listing~\ref{isabelle:follow}.
Again, we use a locale to formalise this function, and that locale fixes the function \isa{parent} mapping every vertex to its parent in its respective tree.
Note that the function \isa{follow} is not well-defined for all possible arguments.
In particular, it is only well-defined if the relation between pairs of vertices induced by the function \isa{parent} is a well-founded relation.
This assumption on \isa{parent} is a part of the locale's definition.

\begin{figure}[H]
\begin{lstlisting}[
 language=Isabelle,
 caption={The definition of a function that ascends the search tree towards the root, returning the traversed path.},
 label={isabelle:follow},
 captionpos=b
 ]
definition follow_invar'::"('a ⇒ 'a option) ⇒ bool" where
  "parent_spec parent ≡ wf {(x, y) |x y. (Some x = par y)}"

locale parent = 
  fixes parent::"'a ⇒ 'a option" and 
    parent_rel::"'a ⇒ 'a ⇒ bool"
  assumes parent_rel:
    "follow_invar' parent"

function follow where
  "follow v = 
     (case (parent v) of Some v' ⇒ v # (follow v')
                       | _  ⇒ [v])"
\end{lstlisting}
\end{figure}

Based on that, $\computaltpath$ is formalised as shown in Listing~\ref{isabelle:compAltPath}.
Note that we do not use a while combinator to represent the while loop: instead we formalise it as a recursive function.
In particular, we define it as a recursive function which takes as arguments the variables representing the state of the while loop, namely, the set of examined edges \isa{ex}, the parent function \isa{par}, and the labelling function \isa{flabel}.

\begin{figure}[H]
\begin{lstlisting}[
 language=Isabelle,
 caption={The definition of a function that constructs the search forest, which is the main search procedure of Edmonds' blossom shrinking algorithm. Note: \isa{f(x := v)} denotes the point-wise update of a function in Isabelle/HOL.},
 label={isabelle:compAltPath},
 captionpos=b
 ]
definition if1 where
  "if1 flabel ex v1 v2 v3 r =
     ({v1, v2} ∈ G - ex ∧ flabel v1 = Some (r, Even) ∧
      flabel v2 = None ∧ {v2, v3} ∈ M)"

definition if1_cond where
  "if1_cond flabel ex = 
     (∃v1 v2 v3 r. if1 flabel ex v1 v2 v3 r)" 

definition if2 where 
  "if2 flabel v1 v2 r r' = 
     ({v1, v2} ∈ G ∧ flabel v1 = Some (r, Even) ∧
      flabel v2 = Some (r', Even))"

definition if2_cond where "if2_cond flabel =
   (∃v1 v2 r r'. if2 flabel v1 v2 r r')"

function compute_alt_path::
  "'a set set ⇒ ('a ⇒ 'a option) ⇒ ('a ⇒ ('a × label) option)
     ⇒ (('a list × 'a list) option)"
  where
  "compute_alt_path ex par flabel = 
    (if if1_cond flabel ex then
       let
         (v1,v2,v3,r) = sel_if1 flabel ex;
         ex' = insert {v1, v2} ex;
         ex'' = insert {v2, v3} ex';
         par' = par(v2 := Some v1, v3 := Some v2);
         flabel' = 
           flabel(v2 := Some (r, Odd), v3 := Some (r, Even));
         return = compute_alt_path ex'' par' flabel'
       in
         return
     else if if2_cond flabel then
        let
          (v1,v2,r,r') = sel_if2 flabel; 
          return = 
            Some (parent.follow par v1, parent.follow par v2)
        in
          return
     else
       let
          return = None
       in
         return)"
\end{lstlisting}
\end{figure}

Note that this function is also defined within a locale, shown in Listing~\ref{isabelle:compAltPathLocale}.
That locale assumes nothing but a choice function that picks elements from finite sets.
\begin{figure}[H]
\begin{lstlisting}[
 language=Isabelle,
 caption={Functions on which \isa{compute\_alt\_path} is parameterised.},
 label={isabelle:compAltPathLocale},
 captionpos=b
 ]
locale match = graph_abs G for G+ 
  fixes M
  assumes matching: "matching M" "M ⊆ G"

locale compute_alt_path = match G M + choose sel 
  for G M::"'a set set" and sel::"'a set ⇒ 'a"
\end{lstlisting}
\end{figure}

One last aspect of our formalisation of $\computaltpath$ is how we model nondeterministic choice and selection.
As mentioned earlier we aimed to arrive at a final algorithm with minimal assumptions on functions for nondeterministic computation.
We thus implement all needed nondeterministic computation using the basic assumed nondeterminitic choice function.
Listing~\ref{isabelle:compAltPath} shows, as an example, how we nondeterministically choose the vertices in the first execution path of the while-loop (i.e.\ the path ending on Line~\ref{compAltPath:ifCondAEnd}).

\begin{figure}[H]
\begin{lstlisting}[
 language=Isabelle,
 caption={The definition of a function that nondeterministically chooses
          vertices and a root that satisfy the conditions of the first execution 
          branch of the while loop.},
 label={isabelle:compAltPath},
 captionpos=b
 ]
definition 
  "sel_if1 flabel ex = 
     (let es = 
        D ∩ {(v1,v2)| v1 v2. {v1,v2} ∈ (G - ex) ∧
                           (∃r. flabel v1 = Some (r, Even)) ∧
                           flabel v2 = None ∧ v2 ∈ Vs M};
         (v1,v2) = sel_pair es;
         v3 = sel (neighbourhood M v2);
         r = fst (the (flabel v1)) 
     in (v1,v2,v3,r))"

lemma sel_if1_works:
  assumes "if1_cond flabel ex"
          "(v1, v2, v3, r) = sel_if1 flabel ex"
  shows "if1 flabel ex v1 v2 v3 r"
\end{lstlisting}
\end{figure}




The functional correctness theorem of Algorithm~\ref{alg:compAltPath}, on the proof of which we spend the rest of this section, is stated as follows.
\begin{mythm}
\label{thm:compAltPathCorrect}
$\computaltpath$ satisfies Specification~\ref{spec:computaltpath}.
\end{mythm}

In the rest of this section we describe the invariants needed to prove this algorithm for searching for augmenting paths is correct (\Cref{subsec:invars}) and then the total correctness proof of the algorithm (\Cref{subsec:compAltPathTerminates}).

\subsection{Loop Invariants}
\label{subsec:invars}

Proving Theorem~\ref{thm:compAltPathCorrect} involves reasoning about a while-loop using loop invariants.
Nonetheless, since the while-loop involves a large number of variables in the state, and those variables represent complex structures, e.g.\ $\parent$, the loop invariants capturing the interactions between all those variables are extensive.
During the development of the formal proof, we have identified the following loop-invariants to be sufficient to prove Theorem~\ref{thm:compAltPathCorrect}.
\newcommand{\rootvertex}{\ensuremath{r}}
\newcommand{\none}{\text{None}}

\begin{myinvar}
\label{invar:altLabelsInvarA} For any vertex $\vertexgen$, if for some $r$, $\lab\;\vertexgen = \langle r, \even\rangle$, then the vertices in the list $\follow{\parent}\vertexgen$ have labels that alternate between $\langle r, \even\rangle$ and $\langle r, \odd\rangle$.
\end{myinvar}

\begin{myinvar}
\label{invar:altLabelsInvarB} For any vertex $\vertexa$, if for some $r$ and some $l$, we have $\lab\;\vertexa=\langle r,l\rangle$, then the vertex list $\vertexa\vertexb\dots\vertexgen_n$ returned by $\follow{\parent} \vertexa$ has the following property: if $\lab\;\vertexgen_i = \langle r, \even\rangle$ and $\lab\;\vertexgen_{i+1} = \langle r, \odd\rangle$, for some $r$, then $\{\vertexgen_i,\vertexgen_{i+1}\}\in\matching$, otherwise, $\{\vertexgen_i,\vertexgen_{i+1}\}\not\in\matching$.
\end{myinvar}

\begin{myinvar}
\label{invar:parentSpec} The relation induced by the function $\parent$ is well-founded.
\end{myinvar}

\begin{myinvar}
\label{invar:flabelInvar} For any $\{\vertexa,\vertexb\}\in\matching$, $\lab\;\vertexa=$ None iff $\lab\;\vertexb=$ None.
\end{myinvar}

\begin{myinvar}
\label{invar:flabelParInvar} For any $\vertexa$, if $\lab\;\vertexa=$ None then $\parent\;\vertexb\neq\vertexa$, for all $\vertexb$.
\end{myinvar}

\begin{myinvar}
\label{invar:lastNotMatchedInvar} For any $\vertexgen$, if $\lab\;\vertexgen\neq$ None, then $\last\;(\follow{\parent} \vertexgen)\not\in\vertices(\matching)$.
\end{myinvar}

\begin{myinvar}
\label{invar:lastEvenInvar} For any $\vertexgen$, if $\lab\;\vertexgen\neq\none$, then $\lab\;(\last\;(\follow{\parent} \vertexgen))=\langle r, \even\rangle$, for some $r$.
\end{myinvar}

\begin{myinvar}
\label{invar:matchedExaminedInvar} For any $\{\vertexa,\vertexb\}\in\matching$, if $\lab\; \vertexa\neq$ None, then $\{\vertexa,\vertexb\}\in\examined$.
\end{myinvar}

\begin{myinvar}
\label{invar:distinctInvar} For any $\vertexgen$, $\follow{\parent} \vertexgen$ is a simple path w.r.t.\ $\graph$.
\end{myinvar}

\begin{myinvar}
\label{invar:flabelInvarB} For any $\{\vertexa,\vertexb\}\in\matching$, $\lab\; \vertexa = \langle r, \even\rangle$ iff $\lab\; \vertexb = \langle r, \odd\rangle$.
\end{myinvar}

\begin{myinvar}
\label{invar:examineHaveOddVertsInvar} For all $\edge\in\examined$, there are $\vertexgen\in\edge$ and $\rootvertex$ s.t.\ $\lab\; \vertexgen = \langle r, \odd\rangle$.
\end{myinvar}

\begin{myinvar}
\label{invar:unexaminedMatchedUnlabelledInvar} For all $\edge\in\graph\setminus\examined$, if $\vertexgen\in\edge$, then $\lab\; \vertexgen = \none$.
\end{myinvar}

\begin{myinvar}
\label{invar:finiteOddsInvar} The set $\{\vertexgen\mid \exists \rootvertex.\;\lab\; \vertexgen = \langle \rootvertex,\odd\rangle\}$ is finite.
\end{myinvar}

\begin{myinvar}
\label{invar:oddLabelledVertsNumInvar} $\card{\{\vertexgen\mid \exists \rootvertex.\; \lab \;\vertexgen = \langle 
\rootvertex,\odd\rangle\}} = \card{\matching \cap \examined}$.
\end{myinvar}

\begin{myinvar}
\label{invar:unlabelledVertsMatchedInvar} For all $\vertexgen\in\vertices(\graph)$, if $\lab\; \vertexgen = \none$, then there is $\edge\in\graph\setminus\examined$ s.t.\ $\vertexgen\in\edge$.
\end{myinvar}

\begin{myinvar}
\label{invar:oddThenMatchedExaminedInvar} For all $\vertexgen\in\vertices(\graph)$, if $\lab\; \vertexgen = \langle\rootvertex,\odd\rangle$, then there is $\edge\in\graph\cap\examined$ s.t.\ $\vertexgen\in\edge$.
\end{myinvar}

\noindent Proofs of those invariants require somewhat complex reasoning: they involve interactions between induction (e.g.\ well-founded induction on $\parent$) and the evolution of the 'program state', i.e.\ the values of the variables as the while-loop progresses with its computation.
We describe one of those formal proofs in some detail below to give the reader an idea of how we did those proofs.

\noindent\textit{Further Notation.} In an algorithm, we refer to the value of a variable $x$ 'after' executing line $i$ and before executing line $i+1$ with $\varAtLine{x}{i}$.

\newcommand{\currentIterVar}[1]{\varAtLine{#1}{\ref{compAltPath:iterBegin}}}
\renewcommand{\nextIterVar}[1]{\varAtLine{#1}{\ref{compAltPath:ifCondAEnd}}}
\newcommand{\retAVar}[1]{\varAtLine{#1}{\ref{compAltPath:retA}}}
\newcommand{\retBVar}[1]{\varAtLine{#1}{\ref{compAltPath:retB}}}

\begin{proof}[Proof sketch of Invariant~\ref{invar:altLabelsInvarA}]
The algorithm has only one execution branch where it continues iterating, namely, when the condition on Line~\ref{compAltPath:ifCondA} holds.
We show that if Invariant~\ref{invar:altLabelsInvarA} holds for $\currentIterVar{\lab}$ and $\currentIterVar{\parent}$, then it holds for $\nextIterVar{\lab}$ and $\nextIterVar{\parent}$.
In particular, we need to show that if, for any vertex $\vertexgen$, $\nextIterVar{\lab}\;\vertexgen = \langle\rootvertex,\even\rangle$, then the labels assigned by $\nextIterVar{\lab}$ to vertices of $\follow{\nextIterVar{\parent}}{\vertexgen}$ alternate between $\langle\rootvertex,\even\rangle$ and $\langle\rootvertex,\odd\rangle$.
The proof is by induction on $\follow{\nextIterVar{\parent}}{\vertexgen}$.
We have the following cases, two base cases and one step case.
\begin{mycase}[$\follow{\nextIterVar{\parent}} \vertexgen = \emptyset$]
This case is trivial.
\end{mycase}
\begin{mycase}[$\follow{\nextIterVar{\parent}} \vertexgen = \rvertexgen$, for some $\rvertexgen$]
This case is also trivial since the $\follow{\nextIterVar{\parent}} \vertexgen$ has no edges.
\end{mycase}
\begin{mycase}[$\follow{\nextIterVar{\parent}} \vertexgen = \rvertexa\rvertexb\cat\path$, for some $\rvertexa$ and $\rvertexb$]
\label{case:lenGeTwo}
The proof can be performed by the following case analysis.
\begin{mysubcase}[$\rvertexb={\nextIterVar{\lvertexb}}$]We further analyse the following two cases.
\label{case:RAeqLCandRBEqLB}
\begin{mysubsubcase}[$\rvertexa=\nextIterVar{\lvertexc}$]
First, we have that $\vertexgen=\rvertexa$ from the definition of $\follow$ and from the assumption of Case~\ref{case:lenGeTwo}.
This together with the assumption of Case~\ref{case:RAeqLCandRBEqLB} imply that $\vertexgen=\nextIterVar\lvertexc$.

We also have that $\{\nextIterVar\lvertexb,\nextIterVar\lvertexc\}\in\matching$, from Line~\ref{compAltPath:chooseVC}, and the fact that neither $\lvertexb$ and $\lvertexc$ change between Lines~\ref{compAltPath:addex}-\ref{compAltPath:ifCondAEnd}.
Also note that from Line~\ref{compAltPath:ifCondA}, we have that $\lab\;\nextIterVar\lvertexb=\none$.
This, together with Invariant~\ref{invar:flabelInvar}, imply that $\lab\;\nextIterVar\lvertexc=\none$.
Thus, from Invariant~\ref{invar:flabelParInvar}, we have that $\parent\;\rvertexgen\neq\vertexc$, for any $\rvertexgen$.
Thus, $\{\nextIterVar\vertexc,\nextIterVar\vertexb\}\cap\follow\;\currentIterVar{\parent}{\lvertexa}=\emptyset$.

From Line~\ref{compAltPath:whileCond} we know that $\lab\;\nextIterVar\lvertexa=\langle\rootvertex,\even\rangle$, for some $\rootvertex$.
Thus, since Invariant~\ref{invar:altLabelsInvarA} holds at Line~\ref{compAltPath:iterBegin}, we know that the labels of $\follow\;\currentIterVar{\parent}\;\nextIterVar\lvertexa$ alternate w.r.t.\ $\currentIterVar{\lab}$.
Since Lines~\ref{compAltPath:iterBegin}-\ref{compAltPath:ifCondAEnd} imply that $\currentIterVar{\parent}\;\lvertexgen = \nextIterVar{\parent}\;\lvertexgen$ and $\currentIterVar{\lab}\;\lvertexgen = \nextIterVar{\lab}\;\lvertexgen$, for all $\vertexgen\not\in\{\nextIterVar\lvertexb,\nextIterVar\lvertexc\}$, and since $\{\nextIterVar\vertexb,\nextIterVar\vertexc\}\cap\follow\;\currentIterVar{\parent}{\lvertexa}=\emptyset$, then we have that $\follow\;\nextIterVar{\parent}\;\nextIterVar\lvertexa$ alternate w.r.t.\ $\nextIterVar{\lab}$.
This, together with the assignments at Lines~\ref{compAltPath:labVB}-\ref{compAltPath:ifCondAEnd} imply that $\follow\;\nextIterVar{\parent}\;\nextIterVar\lvertexc$ alternate w.r.t.\ $\nextIterVar{\lab}$, which finishes our proof.
\end{mysubsubcase}
\begin{mysubsubcase}[$\rvertexa\neq\nextIterVar\lvertexc$]
From the assumption of Case~\ref{case:lenGeTwo}, we have that $\currentIterVar{\parent}\;\rvertexa=\nextIterVar\lvertexb$, which is a contradiction from Invariant~\ref{invar:flabelParInvar} and the condition in Line~\ref{compAltPath:ifCondA}.
\end{mysubsubcase}
\end{mysubcase}
\begin{mysubcase}[$\rvertexa=\nextIterVar\lvertexb$ and $\rvertexb=\nextIterVar\lvertexa$]
This case is implied by Case~\ref{case:RAeqLCandRBEqLB} since $\follow\;\nextIterVar{\parent}\;\nextIterVar\lvertexc=\nextIterVar\lvertexc\cat\follow\;\nextIterVar{\parent}\;\nextIterVar\lvertexb$.
\end{mysubcase}
\begin{mysubcase}[$\{\rvertexa,\rvertexb\}\cap\{\nextIterVar\lvertexb\}=\emptyset$]
\label{case:vBnin}
We perform the following case analysis.
\begin{mysubsubcase}[$\nextIterVar\lvertexc\in\follow{\nextIterVar{\parent}}{\lvertexgen}$]
First, note that, from Line~\ref{compAltPath:ifCondAEnd}, we have that $\nextIterVar{\parent}\;\nextIterVar\lvertexc=\nextIterVar\lvertexb$.
Thus, if $\nextIterVar\lvertexc=\rvertexa$, then we have that $\nextIterVar\lvertexb=\rvertexb$, which is a contradiction from the assumption of Case~\ref{case:vBnin}.
Thus, we have that $\nextIterVar\lvertexc\in\rvertexb\cat\path$.
Note that from Invariants~\ref{invar:flabelInvar} and \ref{invar:flabelParInvar}, Lines~\ref{compAltPath:ifCondA} and \ref{compAltPath:chooseVC}, we have that $\currentIterVar{\parent}\;\lvertexgen\neq\nextIterVar\lvertexc$, for all $\lvertexgen$.
Thus, we have a contradiction.
\end{mysubsubcase}
\begin{mysubsubcase}[$\nextIterVar\lvertexc\not\in\follow{\nextIterVar{\parent}}{\lvertexgen}$]
Note that from Invariant~\ref{invar:flabelParInvar}, and Lines~\ref{compAltPath:ifCondA} and \ref{compAltPath:chooseVC}, we have that, if there is $\rvertexgen\in\rvertexa\rvertexb\cat\path$ s.t.\ $\nextIterVar{\parent}\;\rvertexgen$, then $\rvertexgen=\nextIterVar\lvertexc$, which is a contradiction.
Thus, $\lvertexb\not\in\rvertexa\rvertexb\cat\path$.
Thus for any $\rvertexgen\in\rvertexa\rvertexb\cat\path$, we have that $\currentIterVar{\parent}\;\rvertexgen = \nextIterVar{\parent}\;\rvertexgen$ and $\currentIterVar{\lab}\;\rvertexgen = \nextIterVar{\lab}\;\rvertexgen$.
This finishes our proof, since Invariant~\ref{invar:altLabelsInvarA} holds for $\currentIterVar{\parent}$ and $\currentIterVar{\lab}$.
\end{mysubsubcase}
\end{mysubcase}
\end{mycase}
\end{proof}

\subsection{Total Correctness Proof} 
\label{subsec:compAltPathTerminates}

Below we describe in some detail our formal total correctness proof of the search algorithm.
Although the algorithm has been treated by numerous authors~\cite{edmond1965blossom,LEDAbook,KorteVygenOptimisation,schrijverBook}, we believe that the following proof is more detailed than any previous exposition.
\begin{mylem}
\label{lem:compAltPathTerminates}
Algorithm~\ref{alg:compAltPath} always terminates.
\end{mylem}
\begin{proof}[Proof sketch]
The termination of the algorithm is based on showing that $\card{\graph\setminus\examined}$ decreases with every iteration of the while loop.
\end{proof}

\providecommand{\property}{}

\begin{mylem}
\label{lem:compAltPathSound}
If Algorithm~\ref{alg:compAltPath} returns two paths then they satisfy properties \ref{property:simplePath}-\ref{property:firstNinMatching}.
\end{mylem}
\begin{proof}[Proof sketch]
First, from the definition of the algorithm, we know that the algorithm returns (at Line~\ref{compAltPath:retA}) the two lists $\follow{\retAVar{\parent}}{\retAVar\lvertexa}$ and $\follow{\retAVar{\parent}}{\retAVar\lvertexb}$.
Invariant~\ref{invar:parentSpec} implies that $\follow\;\currentIterVar{\parent}$ is well-defined for any vertex in the Graph, and thus $\follow{\retAVar{\parent}}{\retAVar\lvertexa}$ and $\follow{\retAVar{\parent}}{\retAVar\lvertexb}$ are both well-defined.
We now show that they satisfy the properties \ref{property:simplePath}-\ref{property:firstNinMatching}.
\begin{itemize}
\item \ref{property:simplePath} follows from Invariant~\ref{invar:distinctInvar} and since $\follow{\currentIterVar{\parent}}$ is well-defined.

\item \ref{property:altPath} follows from Invariants~\ref{invar:altLabelsInvarA} and \ref{invar:altLabelsInvarB} and since $\follow{\currentIterVar{\parent}}$ is well-defined.

\item From Line~\ref{compAltPath:chooseVAVB}, we have that $\currentIterVar{\lab}\;\retAVar\lvertexa=\even$.
From Invariant~\ref{invar:lastEvenInvar} and since $\follow{\currentIterVar{\parent}}$ is well-defined, we have that $\currentIterVar{\lab}(\last(\follow{\currentIterVar{\parent}}{\retAVar\lvertexa}))=\even$.
Since $\currentIterVar{\lab}=\retAVar{\lab}$ and $\currentIterVar{\parent}=\retAVar{\parent}$, then $\retAVar{\lab}(\last(\follow{\retAVar{\parent}}{\retAVar\lvertexa}))=\even$.
Also, from Invariant~\ref{invar:altLabelsInvarA}, we have that the vertices in $\retAVar{\lab}(\last(\follow{\retAVar{\parent}}{\retAVar\lvertexa}))$ alternate between labels of $\even$ and $\odd$.
From the properties of alternating lists, we know that if vertices of a list alternate w.r.t.\ a predicate (in this case $\even/\odd$), and the first and last vertex satisfy the same predicate (here $\even$), then the length of this list is odd.
Thus, the length of $\path_1$ is odd.
Similarly, we show that the length of $\path_2$ is odd.
This gives us \property~\ref{property:oddLength}.

\item \ref{property:lastNinMatchingA} and \ref{property:lastNinMatchingB} follow from Invariant~\ref{invar:lastNotMatchedInvar}.

\item \ref{property:firstVtxsEdge} follow from Line~\ref{compAltPath:chooseVAVB}.

\item From Invariant~\ref{invar:altLabelsInvarA}, we have that, since $\retAVar{\lab}\;\retAVar\vertexa=\even$, then the label of the vertex occurring after $\retAVar\lvertexa$ in $\follow{\retAVar{\parent}}{\retAVar\lvertexa}$, call it $\rvertexa$, is labelled as $\odd$.
From Invariant~\ref{invar:altLabelsInvarB}, we thus have $\{\retAVar\lvertexa,\rvertexa\}\in\matching$.
Similarly, we have that $\{\retAVar\lvertexb,\rvertexb\}\in\matching$, where $\rvertexb$ is the vertex occurring in $\follow{\retAVar{\parent}}{\retAVar\lvertexb}$ after $\retAVar\vertexb$.
We thus have that $\{\retAVar\lvertexa,\retAVar\lvertexb\}\not\in\matching$, since no two matching edges can be incident to the same vertex, meaning that we have \ref{property:firstNinMatching}.
\end{itemize}
\end{proof}

\begin{mylem}
\label{lem:compAltPathComplete}
If there are two paths satisfying properties \ref{property:simplePath}-\ref{property:firstNinMatching}, then Algorithm~\ref{alg:compAltPath} returns two paths.
\end{mylem}

\newcommand{\osc}{\text{OSC}}
\newcommand{\capacity}{\text{cap}}

The proof of this lemma depends on showing that we can construct a certificate showing that no such paths exist, if the algorithm returns at Line~\ref{compAltPath:retB}.
The certificate is an odd set cover, defined as follows.
For a set $s\subseteq\vertices(\graph)$, s.t.\ $\card{s}=2k+1$, for some $k$, we define the \emph{capacity} of $s$ as follows:
\[
\capacity(s) = \begin{cases} 1 & \text{ if } k=0\\
                             k & \text{ otherwise}.
               \end{cases}
\]
For a set of edges $\edges$, we say $s$ covers $\edges$ iff $s\cap \edge\neq\emptyset$, for each $\edge\in\edges$, and $k = 0$.
Otherwise, $s$ covers $\edges$ iff $\bigcup\edges \subseteq s$.
A set of sets $\osc$ is an \emph{odd set cover} of a graph $\graph$ iff for every $s\in\osc$, we have that $\card{s}$ is odd and that for every $\edge\in\graph$, there is $s\in\osc$ s.t.\ $s$ covers $\edge$.
We have the following standard property of odd set covers.
\begin{myprop}
\label{prop:OSCOptimalMatching}
If $\osc$ is an odd set cover for a graph $\graph$, then, if $\matching$ is a matching w.r.t.\ $\graph$, we have that $\card{\matching}\leq\card{\osc}$.
\end{myprop}

\begin{mylem}
\label{lem:augPathPropertiesExist}
Given a graph $\graph$ and a matching $\matching$ w.r.t.\ $\graph$, if there is an augmenting path w.r.t.\ $\langle\graph,\matching\rangle$, then are two paths satisfying properties \ref{property:simplePath}-\ref{property:firstNinMatching}.
\end{mylem}
\begin{proof}[Proof sketch]
Let the augmenting path be called $\path$.
First, note that any augmenting path has at least three edges.
Thus, there must be a $\rvertexa$, $\rvertexb$, $\rvertexc$, and $\path_2$, s.t.\ $\path = \rvertexa\rvertexb\rvertexc\cat\path_2$.
Two paths that are the required witness are $\rvertexc\rvertexb\vertexa$ and $\reverse(\path_2)$.
\end{proof}

\begin{mylem}
\label{lem:buildOSC}
If Algorithm~\ref{alg:compAltPath} returns at Line~\ref{compAltPath:retB}, then there is an odd set cover $\osc$ for $\graph$ and $\card{\osc}=\card{\matching}$.
\end{mylem}
\begin{proof}[Proof sketch]
Let $\osc\equiv\{\{\lvertexgen\}\mid\retBVar{\lab}\;\lvertexgen=\odd\}$.
First, we have that $\osc$ is an odd set cover for $\examined$, from the definition of odd set covers and from Invariant~\ref{invar:examineHaveOddVertsInvar}.
Second, since no two vertices in an edge can have the same label of $\odd$ or $\even$, we have that for any $\edge\in\graph$, there is $\lvertexgen\in\edge$ s.t.\ $\retBVar{\lab}\;\lvertexgen=\none$ or $\retBVar{\lab}\;\lvertexgen=\odd$.
The rest of the proof is via the following case analysis.
\begin{mycase}[$\card{\matching}=\card{\examined}$]
From the case assumption and Invariant~\ref{invar:unlabelledVertsMatchedInvar}, we have that
$\retBVar{\lab}\;\lvertexgen\neq\none$ holds for any $\lvertexgen\in\vertices(\graph)$.
Thus, every $\edge\in\graph$ has $\lvertexgen\in\edge$ s.t.\ $\retBVar{\lab}\;\lvertexgen=\odd$.
Accordingly, $\osc$ is an odd set cover for $\graph$.
From the case assumption, in addition to Invariant~\ref{invar:oddLabelledVertsNumInvar}, we have that $\card{\osc}=\card{\matching}$.
Also, since, for every $s\in\osc$, we have that $\card{s}=1$, Proposition~\ref{prop:OSCOptimalMatching} implies that $\matching$ is a maximum cardinality matching.
Thus, Lemma~\ref{lem:augPathPropertiesExist} finishes our proof.
\end{mycase}
\begin{mycase}[$\card{\matching\setminus\examined}=1$]
\label{case:matchingExaminedDiffA}
From the case assumption, there is $\rvertexa$ and $\rvertexb$ s.t.\ $\{\rvertexa,\rvertexb\}\in\matching\setminus\examined$.
Since $\{\rvertexa,\rvertexb\}\not\in\examined$ and from Invariant~\ref{invar:unexaminedMatchedUnlabelledInvar}, we have that $\retBVar{\lab}\;\rvertexa=\retBVar{\lab}\;\rvertexb=\none$.
Let $\osc'$ be $\{\{\rvertexa\}\}\cup\osc$.
\begin{mynote}
For any $\edge'\in\graph\setminus\examined$, there exists $\rvertexgen\in\edge'$ s.t.\ $\{\rvertexgen\}\in\osc'$.
\end{mynote}
\begin{proof}[Proof sketch]
We perform the proof by the following case analysis.
\begin{mysubcase}[$\edge'\in\matching$]
In this case, we have that $\edge'=\{\rvertexa,\rvertexb\}$ from the assumption of Case~\ref{case:matchingExaminedDiffA}.
Our witness $\rvertexgen$ is thus $\rvertexa$.
\end{mysubcase}
\begin{mysubcase}[$\edge'\not\in\matching$]
\label{case:edgePrimeNinMatching}
There must be a $\rvertexgen'\in\edge'$ s.t.\ $\retBVar{\lab}\;\rvertexgen'=\none$ or $\retBVar{\lab}\;\rvertexgen'=\odd$.
We perform the following case analysis.
\begin{mysubsubcase}[$\retBVar{\lab}\;\rvertexgen'=\odd$]
In this case, we have that $\{\rvertexgen'\}\in\osc\subseteq\osc'$, by definitions of $\osc$ and $\osc'$, which finishes our proof.
\end{mysubsubcase}
\begin{mysubsubcase}[$\retBVar{\lab}\;\rvertexgen'=\none$]
\label{case:uPrimeUnlabelled}
There must be $\rvertexgen''$ s.t.\ $\edge'=\{\rvertexgen',\rvertexgen''\}$.
We finish the proof by the following case analysis.
\begin{mysubsubsubcase}[$\retBVar{\lab}\;\rvertexgen''=\none$]
In this case, from the for loop at Line~\ref{compAltPath:labelUnmatched} and from the assumptions of Case~\ref{case:uPrimeUnlabelled}, we have that $\{\rvertexgen',\rvertexgen''\}\in\matching$, which contradicts the assumption of Case~\ref{case:edgePrimeNinMatching}.
\end{mysubsubsubcase}
\begin{mysubsubsubcase}[$\retBVar{\lab}\;\rvertexgen''=\odd$]
In this case, $\rvertexgen''$ is the required witness.
\end{mysubsubsubcase}
\begin{mysubsubsubcase}[$\retBVar{\lab}\;\rvertexgen''=\even$]
This case leads to a contradiction, since it violates the termination assumption of the while loop at Line~\ref{compAltPath:whileCond}.
\end{mysubsubsubcase}
\end{mysubsubcase}
\end{mysubcase}
\end{proof}
From the above note, and since $\osc$ is an odd set cover for $\examined$, we have that $\osc'$ is an odd set cover for $\graph$.
Now, we focus on the capacity of $\osc'$.
We have the following
\begin{flalign*}
\capacity(\osc') =&\; \card{\osc'}& \text{ since for all $s\in\osc'$, $\card{s}=1$} \\
             =&\; \card{\osc} + 1& \text{ by definition}\\
             =&\; \card{\matching\cap\examined} + \card{\matching\setminus\examined}& \text{from Invariant~\ref{invar:oddLabelledVertsNumInvar}}\\
             =&\; \card{(\matching\cap\examined)\cup(\matching\setminus\examined)} &\\
             =&\; \card\matching. &
\end{flalign*}
This finishes our proof.
\end{mycase}
\begin{mycase}[$2 \leq \card{\matching\setminus\examined}$]
From the case assumption, there must be $\rvertexa$ and $\rvertexb$, and $\matching'$, s.t.\ $\matching\setminus\examined=\{\{\rvertexa,\rvertexb\}\}\cup\matching'$.
First, note that $\{\rvertexa,\rvertexb\}\cap\matching'=\emptyset$, since $\matching'$ is also a matching and  since $\{\rvertexa,\rvertexb\}\not\in\matching'$.
Let $\osc'$ denote $\{\{\rvertexa\},\{\rvertexb\}\cup\vertices(\matching')\}\cup\osc$.
We first show the following.
\begin{mynote}
For any $\edge'\in\graph\setminus\examined$, there exists $s\in\osc'$ s.t.\ $\edge'\subseteq s$ or there exists $\rvertexgen\in\edge'$ s.t.\ $\{\rvertexgen\}\in\osc'$.
\end{mynote}
\begin{proof}[Proof sketch]
Our proof is by case analysis.
\begin{mysubcase}[$\edge'=\{\rvertexa,\rvertexb\}$]
In this case, our proof follows since $\{\rvertexa\}\in\osc'$.
\end{mysubcase}
\begin{mysubcase}[$\edge'\in\matching'$]
In this case, our proof follows since $\{\rvertexa\}\cup\vertices(\matching')\in\osc'$ and $\edge'\in\matching'$.
\end{mysubcase}
\begin{mysubcase}[$\edge'\not\in\matching$]
This is a similar case analysis to Case~\ref{case:edgePrimeNinMatching}.
\end{mysubcase}
\end{proof}
Based on the above note, and since $\osc$ is an odd set cover for $\examined$, we have that $\osc'$ is an odd set cover for $\graph$.
Finally, we consider the capacity of $\osc'$.
\begin{flalign*}
\capacity(\osc') =&\; \capacity(\osc) + \capacity(\{\{\rvertexa\},\{\rvertexb\}\cup\vertices(\matching')\}) & \text{ since $\bigcup\{\{\rvertexa\},\{\rvertexb\}\cup\vertices(\matching')\}$}\\
                 & & \text{and $\bigcup\osc$ are disjoint}\\
             =&\; \capacity(\osc) + \capacity(\{\{\rvertexa\}\})\; + & \text{ since $\rvertexa\notin\bigcup\{\{\rvertexb\}\cup\vertices(\matching')\}$}\\
              &\;\;\;\;\;\;\;\;\;\;\;\;\;\;\;\; \capacity(\{\{\rvertexb\}\cup\vertices(\matching')\}) & \\
             =&\; \card{\osc} + 1 + \card{\matching'} & \text{ by definition of capacity}\\
             =&\; \card{\matching\cap\examined} + \card{\matching\setminus\examined} & \text{ from Invariant~\ref{invar:oddLabelledVertsNumInvar}}\\
             =&\; \card{(\matching\cap\examined)\cup(\matching\setminus\examined)} & \\
             =&\; \card\matching. &
\end{flalign*}
\end{mycase}
\begin{mycase}[$\matching\subseteq\examined$]
In this case, we have that for any $\edge\in\matching$, there is $\rvertexgen$ s.t.\ $\{\rvertexgen\}\in\osc$.
Also, for any $\edge\in\graph\setminus\examined$, there is $\rvertexgen$ s.t.\ $\{\rvertexgen\}\in\osc$, using a case analysis like Case~\ref{case:edgePrimeNinMatching}.
Thus, $\osc$ is an odd set cover for $\graph$.
Since $\matching\subseteq\examined$ and from Invariant~\ref{invar:oddLabelledVertsNumInvar}, we have that $\card{\matching}=\card{\osc}$.
Since for all $s\in\osc$, $\card{s}=1$, we have that $\capacity(\osc)=\card{\osc}=\card{\matching}$, which finishes our proof.
\end{mycase}
\end{proof} 

\begin{proof}[Proof sketch of Lemma~\ref{lem:compAltPathComplete}]
We show the contrapositive, i.e.\ we show that if the algorithm returns no paths (i.e.\ terminate at Line~\ref{compAltPath:retB}), then there is not a path satisfying properties \ref{property:simplePath}-\ref{property:firstNinMatching}.
If the algorithm terminates at Line~\ref{compAltPath:retB}, based on Lemma~\ref{lem:buildOSC}, we can construct an odd set cover $\osc$, s.t.\ $\capacity(\osc)=\card{\matching}$.
From Proposition~\ref{prop:OSCOptimalMatching}, we have that $\matching$ is a maximum cardinality matching.
Thus there is not an augmenting path w.r.t.\ $\langle\graph,\matching\rangle$, which, in addition to Lemma~\ref{lem:augPathPropertiesExist}, finishes our proof.
\end{proof}

\begin{proof}[Proof sketch of Theorem~\ref{thm:compAltPathCorrect}]
The theorem follows from Lemmas~\ref{lem:compAltPathTerminates}, \ref{lem:compAltPathSound}, and~\ref{lem:compAltPathComplete}.
\end{proof}

\begin{myremark}
We note that, although Edmonds' blossom shrinking algorithm has many expositions in the literature, none of them, as far we have seen, have a detailed proof of the odd set cover construction (Lemma~\ref{lem:compAltPathComplete}).
We also believe that our work is the first that provides a complete set of invariants and a detailed proof of the odd set cover construction.
The closest exposition to our work, and which we use as an initial reference for our work, is the LEDA book by N\"aher and Mehlhorn.
There, as we mentioned earlier, the authors assumed that the right-left direction of Theorem~\ref{thm:quotient} is not needed and, accordingly, they only have invariants equivalent to Invariants~\ref{invar:altLabelsInvarA}-\ref{invar:matchedExaminedInvar}.
We thus needed to devise Invariants~\ref{invar:distinctInvar}-\ref{invar:oddThenMatchedExaminedInvar} from scratch, which are the invariants showing that the algorithm can construct an odd set cover if the search fails to find two paths or a blossom.
Similarly to N\"aher and Mehlhorn's book, we did not find a detailed proof of this construction in other standard textbooks~\cite{KorteVygenOptimisation,schrijverBook}.
\end{myremark}


\todo{State that MFCS had a wrong proof for total correctness?}

\paragraph{Formalisation}

\begin{figure}[H]
\begin{lstlisting}[
 language=Isabelle,
 caption={A formalisation of Invariants~\ref{invar:altLabelsInvarA}~and~\ref{invar:altLabelsInvarB}.},
 label={isabelle:altLabelsInvar},
 captionpos=b
 ]
fun alt_labels_invar where
  "alt_labels_invar flabel r [] = True"
| "alt_labels_invar flabel r [v] =
     (flabel v = Some (r, Even))"
| "alt_labels_invar flabel r (v1 # v2 # vs) = 
     ((if (flabel v1 = Some (r, Even) ∧ 
           flabel v2 = Some (r, Odd)) then
         {v1,v2} ∈ M
       else if (flabel v1 = Some (r, Odd) ∧
                flabel v2 = Some (r, Even)) then
         {v1, v2} ∉ M
       else undefined)
      ∧ alt_labels_invar flabel r (v2 #vs))"
\end{lstlisting}
\end{figure}

The most immediate challenge to formalising the correctness proof of Algorithm~\ref{alg:compAltPath} is to design the formalisation such that it is possible to effectively manage the proofs for the (numerous) invariants we identified, and combine them into proving that the algorithm is totally correct.
Formalising the invariants themselves is relatively straightforward: we model invariants as predicates parameterised over the relevant variables in the algorithm.
For instance, Listing~\ref{isabelle:altLabelsInvar} shows the formal predicate corresponding Invariants~\ref{invar:altLabelsInvarA}~and~\ref{invar:altLabelsInvarB}.

Proving that the invariants hold, on the other hand, is rather involved.
We note that in much of the previous work on program verification the focus was on automatic verification condition generation from the invariants and the algorithm description.
There, one would state one invariant for the loop, which involves all interactions between the state variables.
Most verification conditions for that invariant are then automatically generated, and most of the generated conditions would be discharged automatically.
The ones which are not are either proved externally as lemmas or proved automatically after strengthening the invariant.

In our setting, the first difference is that many of the invariants need substantial abstract mathematical reasoning and this makes automatic discharging of verification conditions impossible in practice (e.g.\ see the proof of Invariant~\ref{invar:altLabelsInvarA}).
Another difference is that having one comprehensive loop-invariant, which is standard for automated methods, would be infeasible here.
This is because the while-loop we consider is complex and the interactions between state variables are more understandably captured in 16 separate invariants.
We thus structure the verification of invariants manually, where we try to prove every invariant independently, and either assume more of the previously identified invariants, or create new invariants as the need arises.
Listing~\ref{isabelle:altLabelsInvarHolds} shows the formal statement of the theorem showing that \isa{alt\_labels\_invar} is indeed an invariant.
It shows that it is preserved in the recursive execution path (i.e.\ the one ending in Line~\ref{compAltPath:ifCondAEnd}).
In it also the other invariants needed to prove \isa{alt\_labels\_invar} are listed as assumptions.
Lastly, we note that we use computation induction to generate the verification conditions and Isabelle/HOL's standard automation to combine the invariants towards proving the final correctness theorem of the algorithm.

\begin{figure}[H]
\begin{lstlisting}[
 language=Isabelle,
 caption={Lemma showing preservation of Invariants~\ref{invar:altLabelsInvarA} and \ref{invar:altLabelsInvarB}.},
 label={isabelle:altLabelsInvarHolds},
 captionpos=b
 ]
lemma alt_labels_invar_pres:
  assumes
    ass: "if1 flabel ex v1 v2 v3 r" and
    invars: 
      "⋀v r lab. 
          flabel v = Some (r, lab) ⟹ 
            alt_labels_invar flabel r (parent.follow par v)"
      "⋀v r.
         flabel v = Some (r, Even) ⟹
           alt_list (λv. flabel v = Some (r, Even))
                    (λv. flabel v = Some (r, Odd))
                    (parent.follow par v)"
      "parent_spec par"
      "flabel_invar flabel" 
      "flabel_par_invar par flabel" and
    inG: "∃lab. flabel' v = Some (r', lab)" and
    flabel':
      "flabel' = (flabel(v2 ↦ (r, Odd), v3 ↦ (r, Even)))" and
    par': "par' = (par(v2 ↦ v1, v3 ↦ v2))"
  shows "alt_labels_invar flabel' r' (parent.follow par' v)"
\end{lstlisting}
\end{figure}

Other than proving that the invariants hold, we perform the rest of the formal proof of correctness of Algorithm~\ref{alg:compAltPath} by computation induction on \isa{compute\_alt\_path}, which is arduous but standard.

\section{Discussion}
\todo{Compare with refinement framework}
\todo{Compare with Korte and Vygen proofs}
Studying combinatorial optimisation from a formal mathematical perspective is a well-trodden path.
Authors previously studied flows~\cite{FordFulkersonMizar,LammichFlows,ScalingIsabelle}, linear programming~\cite{IsabelleSimplexAFP,LPDualityIsabelle,ThiemannLPs}, and online matching~\cite{RankingIsabelle}, to mention a few.
The only previous authors to formally analyse maximum cardinality matching algorithms in general graphs, as far as we are aware, are Alkassar et al.\cite{matchingCertIsabelle}.
They formally verified a certificate checker for maximum matchings.
In their work, the formal mathematical part amounted to formally proving Proposition~\ref{prop:OSCOptimalMatching}, with the focus mainly on verifying an imperative implementation of the checker.

Studying matching has a deep history.
The first results date back to at least the 19th century, when Petersen~\cite{petersenAugPaths} stated Berge's lemma.
Since then, matching theory has been intensely studied.
This is mainly due to the wide practical applications, like Kidney exchange~\cite{EUKidneyMatchingReview}, pupil-school matching, online advertising~\cite{AdWords2007}, etc.

In addition to practical applications, studying the different variants of matching has contributed immensely to the theory of computing.
This includes, for instance, the realisation that polynomial time computation is a notion of efficient computation, which is the underlying assumption of computational complexity theory as well as the theory of efficient algorithms.
This observation was made by Jack Edmonds in his seminal paper describing the blossom shrinking algorithm~\cite{edmond1965blossom}, where he showed that matchings have a rich mathematical structure that could be exploited to avoid brute-force search for maximum cardinality matchings in general graphs.
Other important contributions of studying matchings include the discovery of the primal-dual paradigm~\cite{HungarianMethodAssignment}, the Isolating lemma by Mulmuley et al.~\cite{mulmuleyVazVaz}, the complexity class \#P~\cite{hashPDef}, and the notion of polyhedra for optimisation problems~\cite{EdmondsWeightedMatching}.
All of that makes it inherently interesting to study matching theory and algorithms from the perspective of formal mathematics.

We have presented the first formalisation of the functional correctness argument of Edmonds' blossom shrinking algorithm.
In doing so we built a reasonably rich formal mathematical library on matching and graphs, stated and proved all invariants, and covered all case analyses.
From a formalisation perspective, we believe that tackling Edmonds' blossom shrinking algorithm is highly relevant.
First, the algorithm has historical significance, as mentioned earlier, making studying it from a formal mathematical perspective inherently valuable.
Second, as far as we are aware, the algorithm is conceptually more complex than any efficient (i.e.\ with worst-case polynomial running time) algorithm that was treated formally.
Thus formalising its functional correctness proof further shows the applicability as well as the utility, e.g.\ to come up with new proofs, of theorem proving technology to complex efficient algorithms.

There are other interesting avenues which we will pursue in the future.
The first practical direction is obtaining an efficient verified implementation of the algorithm.
This could be obtained in a relatively straightforward manner by applying standard methods of refinement~\cite{refinementWirth}, either top-down using Lammich's framework~\cite{lammichMonads} or bottom-up using Greenway et al.'s framework~\cite{AutoCorres}.

Another direction is formalising the worst-case running time analysis of an implementation of the algorithm.
There is a number of challenges to doing that.
A naive implementation would be $O(\card{\vertices(G)}^4)$.
A more interesting implementation~\cite{blossomMNInvAckermann} using the union-find data structure can achieve $O(\card{\vertices(G)}\card{G}\alpha(\card{G},\card{\vertices(G)}))$, where $\alpha$ is the inverse Ackermann function.
Achieving that would, in addition to deciding on a fitting methodology for reasoning about running times, need a much more detailed specification of the algorithm.
Factors affecting the running time include: keeping and reusing the information from the search performed by Algorithm~\ref{alg:compAltPath} after shrinking blossoms, and also carefully implementing the shrinking operation without the need to construct the shrunken graph from scratch.

The most interesting future direction is devising a formal correctness proof for the Micali-Vazirani~\cite{MicaliVazirani} algorithm which has the fastest running time for maximum cardinality matching.
That algorithm achieves a running time of $O(\sqrt{\card{\vertices(G)}}\card{G})$ by using shortest augmenting paths in phases, akin to Hopcroft-Karp's algorithm~\cite{HopcroftKarp1973} for maximum bipartite matching, as well as avoiding blossom shrinking altogether.
The most recent pen-and-paper correctness proof~\cite{MVproofthree} of the algorithm was accepted for publication in 2024 after being under development since at least 2014.
Given that all four previous peer-reviewed pen-and-paper proofs of its correctness had major mistakes~\cite{MicaliVazirani,blumMatching,MVproofone,MVprooftwo}, and given the complexity of the recent pen-and-paper correctness proof, a formal correctness proof for the algorithm would be particularly interesting in either confirming the algorithm's correctness once and for all or finding further mistakes.








\bibliography{long_paper}

\end{document}